\documentclass[11pt,a4paper]{article}
\usepackage[margin=2cm, includefoot,heightrounded]{geometry} 

\makeatletter 
\newcommand\amp{{\if\f@series b\usefont{T1}{cmr}{bx}{it}\else\usefont{T1}{cmr}{m}{it}\fi\& }}
\makeatother

\usepackage{microtype}

\usepackage{tikz}
\usetikzlibrary{calc}
\usepackage{graphicx} 
\usepackage{url} 
\usepackage{hyperref} 
\hypersetup{
	pdftitle={Discrete-time systems in quasi-standard form and the h\_6 coalgebra symmetry},
    pdfauthor={Pavel Drozdov and Giorgio Gubbiotti},
    colorlinks,
    citecolor=blue,
    filecolor=black,
    linkcolor=blue,
    urlcolor=blue,
    final=true
} 
\usepackage[
frenchstyle, 
oldstylenums,
]{kpfonts}  
 \usepackage{makecell} 
\usepackage{authblk}

\title{Discrete-time systems in quasi-standard form\\ and the $\h_6$ coalgebra symmetry}
\author[1]{Pavel Drozdov}
\affil[1]{Dipartimento di Scienze Matematiche, Informatiche e Fisiche, Universit\`a degli Studi di Udine, Via delle Scienze 206, 33100 Udine, Italy  \amp  INFN Sezione di Trieste, Via A. Valerio, 2,
34127  Trieste, Italy;  
\href{mailto:drozdov.pavel@spes.uniud.it}{drozdov.pavel@spes.uniud.it}  \newline 
}
\author[2]{Giorgio Gubbiotti}
\affil[2]{Dipartimento di Matematica ``Federigo Enriques'',
Universit\`a degli Studi di Milano, Via C. Saldini 50, 20133 Milano, Italy
\amp INFN Sezione Milano, Via G. Celoria 16, 20133 Milano, Italy;
\href{mailto:giorgio.gubbiotti@unimi.it}{giorgio.gubbiotti@unimi.it} 
}

\date{\today}

\usepackage{enumerate} 

\usepackage{scrextend}  
\usepackage{setspace}

\usepackage[tiny]{titlesec}
\usepackage{titletoc}  
\titleformat{\section}
  {\normalfont \bfseries \scshape }{\thesection. }{0.1em}{}

\titlecontents{section}[0pt]{\addvspace{1ex}}{\thecontentslabel. \normalfont \bfseries \scshape}
{\Large\bfseries}{\hspace{1em plus 1fill}\large\bfseries \contentspage}

\titleformat{\subsection}
{\normalfont \bfseries}{\thesubsection. }{0.1em}{}
\titleformat{\subsubsection}
{\normalfont \small \bfseries}{\thesubsubsection. }{0.1em}{}

\renewcommand{\thesection}{\Roman{section}}

\usepackage{amsthm}
\usepackage{cleveref}
\usepackage{bbding}
\usepackage[color=yellow]{todonotes}

\usepackage{physics}
\usepackage{booktabs}
\usepackage{mathtools}

\newcommand{\galg}{\mathfrak{g}}
\newcommand{\h}{\mathfrak{h}}
\renewcommand{\vec}{\boldsymbol}
\newcommand{\vq}{\vec{q}}
\newcommand{\vQ}{\vec{Q}}
\newcommand{\vp}{\vec{p}}
\newcommand{\vP}{\vec{P}}
\newcommand{\vl}{\vec{\lambda}}
\newcommand{\Z}{\mathbb{Z}}
\newcommand{\R}{\mathbb{R}}
\newcommand{\I}{\mathbb{I}}

\DeclareMathOperator{\Span}{span}
\DeclareMathOperator{\id}{id}
\DeclareMathOperator{\Pol}{Pol}

\theoremstyle{definition}
\newtheorem{definition}{Definition}[section]
\theoremstyle{theorem}
\newtheorem{theorem}{Theorem}[section]

\newtheorem*{theorem*}{Theorem}
\newtheorem{corollary}[theorem]{Corollary}
\newtheorem{lemma}[theorem]{Lemma}
\newtheorem{proposition}[theorem]{Proposition}

\theoremstyle{remark}
\newtheorem{remark}{Remark}[section]

\numberwithin{equation}{section}

\usepackage{amsmath}
\usepackage{mathbbol}
\usepackage{braket}
\usepackage{tensor}

\usepackage[backend=biber,sorting=nyt,maxnames=99,giveninits=true,doi=false,isbn=false,style=alphabetic,url=false,maxalphanames=4,minalphanames=3]{biblatex}
\begin{document}

\maketitle

\begin{abstract}
    \noindent
    In this paper, we characterize all discrete-time systems in
    quasi-standard form admitting coalgebra symmetry with respect to the
    Lie--Poisson algebra $\h_{6}$. The outcome of this study is a family of
    systems depending on an arbitrary function of three variables, playing
    the r\^ole of the potential. Moreover, using a direct search approach,
    we classify discrete-time systems from this family that admit an
    additional invariant at most quadratic in the physical variables. We
    discuss the integrability properties of the obtained cases, their
    relationship with known systems, and their continuum limits.
\end{abstract}

\tableofcontents

\section{Introduction}

Continuous and discrete integrable systems are of fundamental importance in
applications because their behavior can be fully understood, and they usually
form ``universality classes'' for generic systems. See for
instance~\cite{Calogero1991}, where it was shown that many integrable partial
differential equations (PDEs) are obtained from general nonlinear equations
through asymptotic analysis.

For this reason, seeking and constructing algorithmically new integrable
systems is one of the main tasks in the integrable systems theory, and it has
been pursued since its beginning, see, for instance, the work of
Darboux~\cite{Darboux1901}. With a focus on finite-dimensional integrable
systems, \textit{i.e.}\ systems of ordinary differential or difference
equations, many methods to construct integrable systems were developed
throughout the years, see for instance the
reviews~\cite{Hietarinta1987,MillerPostWinternitz2013R}. Most methods, like the
direct construction of the invariants~\cite{Fris1965} or the Jacobi
geometrization method~\cite{Drach1935}, become computationally impractical as
the size of the underlying system increases. So, many methods relying on
higher-level algebraic or geometric structures were devised.

In this paper, we characterize and study a class of discrete-time systems in
$N$-degrees of freedom admitting coalgebra symmetry with respect to the Lie
algebra $\h_{6}$~\cite{Zhang_EtAl1990.Coherent_states_Theory_applications}.
To be more precise, our object of study will be second-order difference
equations
\begin{equation}\label{discrete-system-general-form}
    q_{k} (t+h) =  F_k (\vec q(t), \vec q(t-h)),
    \quad
    \vec{q}(t) = \left( q_{1}(t),\dots,q_{N}(t) \right),
    \quad k=1,\dots,N,
\end{equation}
for an unknown sequence  $\{  \vec{q}_k \}_{k \in \mathbb{Z}} $ in
$\mathbb{R}^N $ for $N \in \mathbb{N} $, $h>0$ is a (fixed) constant, and $t\in
h\Z$.  We follow the classical notation of~\cite{McLachlan1993,
    Suris1994Garnier, Suris1994Symmetric, Suris1994inversesquare} to avoid
double indexing in the formul\ae.  From now on we will call the number $N$  \emph{degrees of freedom} of our system. We will focus on the systems arising as a discrete Euler--Lagrange equations of the following class of discrete Lagrangians~\cite{Logan1973,Gubbiotti_dcov}:
\begin{equation}
    L = \sum_{k=1}^{N} \ell_{k}\left(q_{k}(t+h) q_{k}(t)\right)-
    V(\vec{q}(t)),
    \label{eq:dLgen}
\end{equation}
where the functions $\ell_{k}=\ell_{k}(\xi)$ are smooth and locally invertible functions in a given open domain of $\R$. Throughout this paper, we will refer
this kind of systems as \emph{systems in quasi-standard form}. The term
``systems in quasi-standard form'' was introduced in~\cite{GubLat_sl2} as an
extension of the systems in the standard form used
in~\cite{Suris1989,McLachlan1993}, which is obtained for $\ell_{k}(\xi)=\xi$,
for $k=1,\dots,N$. Particular cases of systems in quasi-standard form appeared
in~\cite{Suris1994inversesquare}.

The discrete Euler--Lagrange equations in the case of the Lagrangian
\eqref{eq:dLgen} take the form: 
\begin{equation}
    \ell_k' \left( q_k(t+h) q_k(t) \right) q_k(t+h)
    +\ell_k' \left( q_k(t) q_k(t-h) \right) q_k(t-h) 
    = \frac{\partial V (\vec q (t))}{\partial q_k (t)}.
    \label{eq:dEL}
\end{equation}
Defining the canonical momenta as $p_k (t) = \ell_k' \left( q_k(t) q_k(t-h)
\right) q_k(t-h),$ $k=1,\dots, N$, the discrete Euler--Lagrange
equations~\eqref{eq:dEL} can be written in canonical form as:
\begin{subequations} \label{eq:dEL-quasistandard}
    \begin{align}
        \ell_k' \left( q_k(t+h) q_k(t) \right) q_k(t+h) + p_k(t)
        = \frac{\partial V (\vec q(t)) }{ \partial q_k (t)},
        \\
        p_k(t+h) = \ell_k' \left( q_k(t+h) q_k(t) \right) q_k(t), 
        \label{eq:dEL-quasistandard-b}
    \end{align}
\end{subequations}
see~\cite{Bruschi_EtAl1991.Integrable_symplectic_maps}.

Our construction will make use of the notion of coalgebra symmetry for
discrete-time systems that was recently introduced by one of the authors
in~\cite{Gubbiotti_EtAl2023.Coalgebra_symmetry_discrete_systems}. We produce
several examples, including some new generalizations, of (super)integrable
discrete systems and quasi-integrable discrete systems in arbitrary degrees of
freedom. This paper extends a previous work~\cite{GubLat_sl2}, where an
analogous classification and study was carried out in the case of the Lie
algebra $\mathfrak{sl}_{2}(\R)$, a Lie subalgebra of $\h_{6}$.

The coalgebra symmetry approach is an algebraic method that was introduced
in~\cite{Ballestreros_EtAl1996.Ndimensional_classical_integrable_systems_Hopf_algebras},
and later developed
in~\cite{Ballesteros_Ragnisco1998.systematic_construction_completely_integrable_Hamiltonians_coalgebras},
to algorithmically construct (classical) Liouville integrable systems in
arbitrary degrees of freedom. The idea behind this approach is to extend dynamical systems from one to $N$ degrees of freedom by interpreting the ``base case'' as a symplectic realization of a Poisson (co)algebra, and then
embedding it into $N$-degrees of freedom through tensor products, and through a coassociative map called the coproduct. The obtained systems in higher degrees of freedom are naturally endowed with constants of motion arising as the images
of the Casimir invariants of the Poisson (co)algebra.

Throughout the years this method has been extended to many other cases, and
many integrable and superintegrable systems have been understood within this
framework, see for instance the review~\cite{Ballesteros_EtAl2009.Super_integrability_coalgebra_symmetry_Formalism_applications}. Additional
applications and examples of the coalgebra symmetry approach include
superintegrable systems defined on non-Euclidean spaces
\cite{BallesterosHerranz2007,Ballesteros_et_al2008PhysD,Ballesteros_et_al2009AnnPhys,
    BallesterosHerranz2009, Ballesteros_et_al2011AnnPhys, Riglioni2013,
PostRiglioni2015}, models with spin-orbital interactions
\cite{Riglioni_et_al2014}, discrete quantum mechanical systems
\cite{LatiniRiglioni2016}, and superintegrable systems related to the
generalized Racah algebra $R(n)$  \cite{DeBie_et_al2021, Latini2019,
Latini_et_al2020embedding, Latini_et_al2021}. We also mention some extensions of
the method given in~\cite{Ballesteros_et_al2002} to comodule algebras, in
\cite{Musso2010loop} to loop coproducts, and last but not least the extension of the method to the discrete setting in~\cite{GubLat_sl2,Gubbiotti_EtAl2023.Coalgebra_symmetry_discrete_systems}.

The structure of the paper is the following: in \cref{sec:background} we give
some background material, including the notion of integrability and
quasi-integrability of difference equations, and we present the coalgebra
symmetry construction for the Lie algebra $\h_{6}$. \Cref{sec:main-result}
contains the main result of this paper: a characterization of all systems in
quasi-standard form admitting coalgebra symmetry with respect to the Lie
coalgebra $\h_6$. In particular, we show that those systems reduce to standard
form, \emph{i.e.} $\ell_{k}(\xi)=\xi$.  In \cref{sec:integrability-props}, we
impose the existence of a polynomial invariant for the system on the generators
of the coalgebra and we isolate six different potentials admitting invariants.
Some potential depends on arbitrary functions, \emph{i.e.} they are
\emph{singular} in a sense that will be specified later.  In \cref{sec:sys}, we
study in detail the integrability properties of these systems in standard form
in an arbitrary number of degrees of freedom, proving that there is one
maximally superintegrable system, three superintegrable systems, and two
quasi-integrable systems. The quasi-integrable systems are the singular ones.
Finally, in \cref{sec:conclusions}, we give some conclusions and an
outlook for future works and potential further directions. We also present
in \cref{tab:resuming} a summary of the systems we found and their properties.

\section{Background material}
\label{sec:background}

In this section, we give a brief outlook of the background material that we need for our discussion. In particular, we will review the notion of
integrability for discrete systems, with a special focus on the formulation of the coalgebra symmetry for discrete systems. Then, we will present the
coalgebra symmetry construction in the case of $\h_{6}$ algebra, demonstrating its universal invariants.

\subsection{Discrete-time integrable systems} 
\label{subsec:def-s, discrete time}

We will need several definitions from integrable systems theory. First of all, for us, an $N$-dimensional system of equations (either continuous or discrete) is \textbf{integrable} if it possesses enough constants of motion (or
invariants), \textit{i.e.} \textit{functionally independent conserved
quantities} $I(t)$ such that along the evolution of the system
 \begin{equation}\label{invariants-general-def-t+h}
 	I (t+h) = I(t). 
 \end{equation}
 In the discrete system case, the functional relation \eqref{discrete-system-general-form} can be interpreted as a \textit{finite-translation map}
 \begin{equation}\label{finite-translation-map} 
 	\Phi_h: x_k (t) \mapsto x_k(t+h). 
 \end{equation}
 We call \eqref{finite-translation-map} a \textbf{map form of the difference equation \eqref{discrete-system-general-form}}. In the map language, \eqref{invariants-general-def-t+h} can be rewritten as 
 \begin{equation}\label{invariants-general-def-pullback}
 	\Phi_h^* I (t)= I (t).
 \end{equation}
 
In the theory of continuous integrable systems, one usually considers
Hamiltonian systems, defined by the time flow of a Hamiltonian function
$\mathcal{H}$ on a symplectic manifold $(\mathcal{M}, \omega)$, or more
generally a \textit{Poisson manifold} $(\mathcal{P}, \{ \, , \, \} )$.  When no
possibility of confusion arises, we will drop the explicit reference to the
symplectic form $\omega$, and to the Poisson bracket $\left\{ , \right\}$.

The analog of this concept in the discrete setting can be understood as follows:
take an initial condition $x(0)\in\mathcal{M}$. Then, its time evolution through a
Hamiltonian flow:
\begin{align}
    \Phi: \  \mathcal{M}& \to \mathcal{M} \notag 
    \\ 
    x(0)	& \mapsto x(t)
\end{align}
leaves invariant the symplectic form: $\Phi^* \omega =\omega $, \emph{i.e.} it
is a \textit{symplectic map}. Iterating such a map, we get a discrete-time
system. So, it is natural to reformulate the notion of integrability for
discrete-time systems in the language of symplectic maps. That is, we have the following definition of integrability:
 
\begin{definition}[Liouville integrability of discrete systems~\cite{Bruschi_EtAl1991.Integrable_symplectic_maps}] 
    Let $\Phi_h$ be a $2N$-dimensional symplectic map, \textit{i.e.} map form
    of a system of difference equations in $N$ degrees of freedom. Hence
    $\Phi_h$ is \textbf{Liouville integrable} if it possesses $N$ functionally
    independent invariants in involution with respect to the associated Poisson
    bracket for any $h>0$. 	
\end{definition}
 
 \begin{remark}
        We recall that a non-degenerate Poisson manifold $\mathcal{P}$ is isomorphic 
        to a symplectic manifold $\mathcal{M}$. Hence, the associated Poisson bracket for $f,g \in \mathcal{C}^\infty( \mathcal{P}) $ is given by 
 	\begin{equation}
 		\{f , g \} = \omega (X_f, X_g),
 	\end{equation}
        where $X_f, X_g $ are the corresponding Hamiltonian vector fields,
        see for instance the textbook~\cite[Chap. 6]{Olver1986}. 
 	
        In the degenerate case, one can slice a Poisson manifold into symplectic
        leaves (with each of them being a symplectic manifold), using the
        $N$-dimensional  \textbf{symplectic realization} 
 	\begin{equation}
 		\mathrm{D}: \ (\mathcal{M}, \omega) \to (\mathcal{P},  \{ \, , \, \} ),
 	\end{equation}
        given, for example, in terms of \textit{Darboux variables} 
        $\{q_i,p_j\} = \delta_{ij} $: 
 	\begin{equation}
            \mathrm{D}: x \mapsto x (q_1, \dots, q_N, p_1, \dots, p_N). 
 	\end{equation}
        The number of invariants to integrate the system now depends on the
        rank of the Poisson bracket. 
 \end{remark}

We also recall a few definitions for special cases of integrable systems: 
\begin{definition}\label{def:S-QMS-MS}
    Let $\Phi_h$ be a $2N$-dimensional symplectic Liouville integrable map. Then:
    \begin{enumerate}[(i)]
            \item If $\Phi_h$ possesses $k$ additional invariants in involution with respect to the Poisson bracket for any $h>0$ (\textit{i.e.} $N+k$ in total), it is said to be a \textbf{superintegrable map}. 
            \item If $\Phi_h$ possesses $N-2$ additional invariants in involution with respect to the Poisson bracket for any $h>0$ (\textit{i.e.} $2N-2$ in total), it is said to be a \textbf{quasi-maximally superintegrable (QMS)} map.\label{def:QMS} 
            \item  If $\Phi_h$ possesses $N-1$ additional invariants in involution with respect to the Poisson bracket for any $h>0$ (\textit{i.e.} $2N-1$ in total), it is said to be a \textbf{maximally superintegrable (MS)} map. \label{def:MS}  
    \end{enumerate}
\end{definition}

We also recall the following definition concerning systems that are ``almost''
integrable:

\begin{definition}\label{def:quasi-integrable}
    Let $\Phi_h$ be a $2N$-dimensional symplectic map. Then $\Phi_{h}$ is
    \textbf{quasi-integrable (QI)} when it possesses $N-1$ commuting
    invariants.  
\end{definition}
 
Finally, we recall the notion of coalgebra symmetry for discrete-time
systems as introduced
in~\cite{Gubbiotti_EtAl2023.Coalgebra_symmetry_discrete_systems}.  The concept
of coalgebra was formulated in quantum group theory
\cite{ChariPressley1994Book,Drinfeld1987}. Precisely, a \emph{coalgebra} is a
pair of objects $(\mathscr{U},\Delta)$ where $\mathscr{U}$ is a unital, associative
algebra and $\Delta\colon \mathscr{U} \rightarrow \mathscr{U}\otimes
\mathscr{U}$ is a \emph{coassociative} map.  That is, $\Delta$ satisfies the
following condition:
\begin{equation}
    (\Delta \otimes \id) \circ \Delta=(\id \otimes \Delta) \circ \Delta
    \iff
    \begin{tikzpicture}[baseline={(0,0)},thick]
        \node (a1) at (0,0){$\mathscr{U}$};
        \node (a2) at ($({2*cos(30)},{2*sin(30)})$) {$\mathscr{U}\otimes\mathscr{U}$};
        \node (a3) at ($({2*cos(30)},{-2*sin(30)})$) {$\mathscr{U}\otimes\mathscr{U}$};
        \node (a4) at (4,0){$\mathscr{U}\otimes\mathscr{U}\otimes\mathscr{U}$};
        \draw[->] (a1) edge node[above left]{$\Delta$} (a2) (a2) edge node[above right]{$\Delta\otimes\id$}(a4) ;
        \draw[->] (a1) edge node[below left]{$\Delta$} (a3) (a3) edge node[below right]{$\id\otimes\Delta$} (a4);
    \end{tikzpicture}
    \label{eq:commdiag}
\end{equation}
and it is an algebra homomorphism from $\mathscr{U}$ to  
$\mathscr{U} \otimes \mathscr{U}$:
\begin{equation}
    \Delta (X \cdot  Y) = \Delta (X)  \cdot \Delta (Y) \qquad \forall \, X, Y \in \mathscr{U} \, .
    \label{eq:homalg}
\end{equation}
The map $\Delta$ is called the \emph{coproduct map}. When there is no possible
confusion on the coproduct map, it is customary to denote the coalgebra simply
by $\mathscr{U}$. Our case of interest are \emph{Poisson coalgebras},
\emph{i.e.} the algebra $\mathscr{U}$ is a commutative associative algebra,
admitting a bilinear derivation $\left\{ \, , \, \right\}$, satisfying the
Jacobi identity, the \emph{Poisson bracket}. Additionally, we also require that
the coproduct map is a \emph{Poisson homomorphism}, with respect to the
standard induced Poisson bracket on $\mathscr{U}\otimes\mathscr{U}$:
\begin{equation}
    \left\{X \otimes Y, W \otimes Z\right\}_{\mathscr{U} \otimes \mathscr{U}}
    \coloneqq 
    \left\{X, W\right\}_{\mathscr{U}} \otimes YZ
    +XW \otimes \left\{Y,Z\right\}_{\mathscr{U}} 
    \, , \quad X,Y,Z,W \in \mathscr{U}.
    \label{eq:pstrucatensa}
\end{equation}

Then, we give the following definition:
\begin{definition}[Coalgebra symmetry for the discrete-time systems, \cite{Gubbiotti_EtAl2023.Coalgebra_symmetry_discrete_systems}]  \label{def-coalg-symmetry-discrete} 
    Let us assume we are given the following:
    \begin{itemize}
        \item  $\Phi_h$ a $2N$-dimensional symplectic map;
        \item $(\mathscr{U}, \Delta)$ a Poisson coalgebra with generators
            $\Set{J_1, \dots, J_k}$;
        \item a symplectic realization $\mathrm{D}$ of $\mathscr{U}$ in $N$ degrees of freedom,
    \end{itemize}
    Then, $\Phi_h$ is said to possess \textbf{coalgebra symmetry} with respect to $(\mathscr{U}, \Delta)$, if for any $N \in \mathbb{N}$ the evolution of the generators in $\mathrm{D}$ is:
        \begin{enumerate}[(i)]
            \item closed in the Poisson coalgebra:
            \begin{equation}
                \Phi_h^* J_i (t) = J_i (t+h) = 
                a_i \left( J_1(t), \dots, J_k(t) \right), \qquad i=1, \dots, k, 
            \end{equation}
            with $a_i \in \mathcal{C}^\infty (\mathscr{U}) $;
            \item $T_h$ is a Poisson map with respect to $\mathscr{U}$: 
            \begin{equation}
                \left\{J_i (t+h),\
                 J_j (t+h) \right\}_\mathscr{U} = \Phi_h \left( \{J_i(t),\ J_j (t) \}_\mathscr{U} \right), \qquad 
                i,j=1,\dots, k;
            \end{equation}
            \item assuming that $\mathscr{U}$ admits $r$ independent Casimir functions $C_1(t), \dots, C_r (t)$, they are preserved as invariants by $\Phi_h$: 
            \begin{equation}
                \Phi_h^*  C_i \left( t \right) \equiv C_i (t+h) = C_i(t), \qquad
                i=1,\dots, r. 
            \end{equation}
            \label{def-iii} 
        \end{enumerate}
\end{definition}

Then, under this assumption it is possible to prove the following theorem:

\begin{theorem}[\cite{Gubbiotti_EtAl2023.Coalgebra_symmetry_discrete_systems}] 
   A $2N$-dimensional symplectic $\Phi_{h}$ map admitting coalgebra symmetry
   with respect to the coalgebra $(\mathscr{U},\Delta)$ admits a set of $2(N-1)r$
   invariants of which $(N-1)r$ are commuting.
   \label{thm:sumup}   
\end{theorem}

\begin{remark}
    The proof of the previous theorem involves the explicit construction of two
    sets of invariants, called the \emph{left and right Casimir functions}.
    Here, for the sake of brevity, we will not present the general formulas,
    but we will show in the next subsection how the construction works in the
    case of interest, \emph{i.e.} the $\h_{6}$ algebra. Even more importantly,
    we observe that \emph{not all invariants} produced with this procedure will
    be functionally independent. So, the number $2(N-1)r$ has to be interpreted
    as an \emph{upper bound} on the number of \emph{bona fide} invariants
    obtainted through the coalgebra symmetry method.
    \label{rmk:casimirs}
\end{remark}

Let us now observe that in~\cite{GubLat_sl2} it was shown how to use the
\Cref{def-coalg-symmetry-discrete} constructively to characterize all systems
in quasi-standard form admitting coalgebra symmetry with respect to the
$\mathfrak{sl}_2 (\mathbb{R})$ Lie algebra.  In this paper, we make the following logical step in developing the theory of coalgebra symmetry for
discrete systems, that is finding an analogous result in the case of the more
general $\mathfrak{h}_6$ algebra, which contains $\mathfrak{sl}_2 (\mathbb{R})$
as a subalgebra. 

\subsection{Coalgebra symmetry application to the two-photon $\mathfrak{h}_6$ Lie algebra}
\label{subsec:h6}

The two-photon Lie algebra, $\mathfrak{h}_{6}$ algebra is spanned by the six generators 
$\vec J \coloneqq   \left\{ K,A_{+},A_{-},B_{+},B_{-},M \right\}$
with the following commutation rules:
\begin{equation}
    \begin{array}{lll}
        [K,A_{+}] = A_{+}, & [K,A_{-}] =-A_{-}, &
        [A_{-},A_{+}] = M,
        \\{}
        [K,B_{+}], = 2 B_{+}, &
        [K,B_{-}] = -2 B_{-} &
        [B_{-},B_{+}] = 4K +2M,
        \\{}
        [A_{+},B_{-}] = -2 A_{-}, &
        [A_{+},B_{+}] = 0,
        &
        [M, \phantom{A} \cdot\ ] = 0,
        \\{}
        [A_{-},B_{+}] = 2 A_{+}, &
        [A_{-},B_{-}] =0.
    \end{array}
    \label{eq:h6comm}
\end{equation}
This is a non-semisimple Lie algebra with a central element $M$.
            
\begin{remark}
    The $\h_{6}$ algebra admits as Lie subalgebras many other notable Lie
    algebras: the already discussed $\mathfrak{sl}_{2}(\R)$; its central
    extension $\mathfrak{gl}_{2}(\R)$; the oscillator algebra $\h_{4}$; and the
    Heisenberg algebra $\h_{3}$. Finally, $\h_{6}$ itself is isomorphic to the
    $(1+1)$ Schr\"odinger Lie algebra~\cite{PhysRevD.5.377,
    niederer1972maximal}.
    \label{rmk:subalg}
\end{remark}

To get a Poisson algebra from a Lie algebra, we consider its \emph{symmetric
algebra}. The symmetric algebra $S\galg$ of a  Lie algebra $\galg$ is
constructed from the polynomials on the dual of a Lie algebra: $S\galg\cong
\Pol (\galg^* )$, where we denoted by $\Pol(\galg^{*})$ the polynomials in the
coordinate functions of $\galg$, see for instance~\cite[Chap.
7]{laurent-gengouxPoissonStructures2013}\footnote{In fact over the real numbers
    this construction extends easily to the \emph{algebra of functions} of
    $\galg$, but for our scopes in this paper it will be enough to consider the
symmetric algebra.}. Indeed, on the coordinate functions, which by abuse of
notation we will still denote with the same letters of the basis of the Lie
algebra, we define a linear Poisson bracket $\{ \ , \ \}$, called the
Lie--Poisson bracket. This Lie--Poisson bracket is induced by the structure
constants of the original Lie algebra $\galg$.  The Lie--Poisson bracket is
explicitly given by:
\begin{equation}
    \{ J_i, \ J_j \} = \sum_k s^k_{ij} J_k,
\end{equation}
where $J_i, J_j$ are the basis of $\galg$, and $s^k_{ij} $ are the structure
constants of Lie algebra $\galg$. The bracket is then extended to all elements of $S\galg$ by the Leibnitz rule.
 
In our case of study, the symmetric algebra $S\h_{6}$ is endowed with the following quartic Casimir element
\cite{Ballesteros_EtAl2009.Ndimensional_integrability_twophoton_coalgebra_symmetry}:    
\begin{equation}
    C_0 = (M B_+ - A_+^2) (MB_- - A_-^2) - 
    (MK- A_+ A_- + M^2/2)^2,
\end{equation}
which can be factorized as: 
\begin{equation}\label{eq:Cash6} 
    C \coloneqq \frac{C_0}{M} = 
        M B_+ B_- 
        - B_+ A_-^2 
        - B_- A_+^2 
        - M \left( K + \frac{M}{2} \right)^2
        + 2 A_- A_+ \left( K + \frac{M}{2} \right),
\end{equation}
since the central element $M$ is a trivial Casimir element.

We recall how to use the coalgebra symmetry approach to $S\h_{6}$ to obtain two
families of $N-2$ new Casimir elements for a given dimension $N \in \mathbb{N}_{>2}$, following the construction made in
\cite{ballesterosTwoPhotonAlgebraIntegrable2001a,
Ballesteros_EtAl2009.Ndimensional_integrability_twophoton_coalgebra_symmetry},
 
First, the Lie--Poisson algebra $S\h_{6}$ can be endowed with a coalgebra
structure using the \textit{primitive coproduct}
\cite{Tjin1992.introduction_quantized_Lie_groups_algebras} $\Delta\colon S\h_6
\to S\h_6 \otimes S\h_6 $, acting as 
 \begin{equation}
 	\Delta(J)= J \otimes 1 + 1 \otimes J, 
 	\qquad
 	J \in \vec J. 
\end{equation}
The primitive coproduct is a Poisson algebra homomorphism. Moreover, for $m>2$
by recursion one can define the \emph{left and right} higher-order $m$th
\emph{coproduct} $\Delta^{(m)}_{R}\colon \h_6 \rightarrow \h_6^{\otimes m}$,
and $\Delta^{(m)}_{L}\colon \h_6 \rightarrow \h_6^{\otimes m}$, acting as:
\begin{equation}\label{higher-order-coproducts}
    \Delta_{L}^{(m)}:=
    \bigl(\overbrace{\id\otimes \id\otimes \dots \otimes \id}^{m-2}\otimes
    \Delta\bigr)\circ\Delta^{(m-1)},
    \quad
    \Delta^{(m)}_{R}:=
        \bigl(\Delta\otimes\overbrace{\id\otimes \id\otimes \dots \otimes \id}^{m-2}\bigr)\circ\Delta^{(m-1)}.
\end{equation}
By induction on $m$, it follows that the $m$th coproducts $\Delta^{(m)}_{L}$ and
$\Delta_{R}^{(m)}$ are also Poisson maps on $\h_6^{\otimes
m}$. 

Now, to discuss integrable systems, we need a realization of $S\h_{6}$ into a
symplectic manifold. It is easy to check that a one-degree-of-freedom symplectic realization
of $S\h_{6}$ is given by:
\begin{equation}
    \begin{aligned}
 	\mathrm{D}\colon S\h_{6} & \to \mathcal{C}^{\infty}(\R^{2}_{(q,p)}) 
 	\\ 
	(M, A_+, A_-, B_+, B_-, K)  	&
 \mapsto
	\big( M, A_+ (q, p) , A_-( q,  p) , B_+ (q, p) , B_-(q, p) , K ( q,  p)    \big)
    \end{aligned}
    \label{eq:h6repr1d}
\end{equation}
where $(q,p)$ are Darboux coordinates, and the explicit expressions 
of the functions in~\eqref{eq:h6repr1d} are given by:
\begin{equation}
    \begin{array}{lll}
        \mathrm{D} (  A_{+}) = \lambda p, &	\mathrm{D}( A_{-} )= \lambda q,
        &
        \mathrm{D} ( K )= qp -\displaystyle\frac{\lambda^{2}}{2},
        \\
        \mathrm{D}( B_{+}) = p^{2}, & 	\mathrm{D} (B_{-} )= q^{2}, & 	\mathrm{D}(M) =
        \lambda^{2},
    \end{array}
    \label{eq:h6rea}
\end{equation}
with $\lambda\in\R\setminus\Set{0}$ a non-vanishing parameter.

So, acting with the primitive coproduct, we obtain the following $N$-degrees of
freedom in Darboux coordinates $(\vec{q},\vec{p})$:
\begin{equation}
    \begin{aligned}
	\mathrm{D}: \quad S\h_{6} & 
	\to 
        \mathcal{C}^{\infty}(\R^{2N}_{(\vec{q},\vec{p})})  
	\\ 
(M, A_+, A_-, B_+, B_-, K) 
& \mapsto
	\big( M, A_+ (\vec q, \vec p) , A_-(\vec q, \vec p) , B_+ (\vec q, \vec p) , B_-(\vec q, \vec p) , K (\vec q, \vec p)    \big),
    \end{aligned}
    \label{Nrealization-map}  
\end{equation}
which is explicitly given by             
\begin{equation} 
    \begin{array}{lll}
        A_{+} = \displaystyle\sum_{i=1}^{N}\lambda_{i} p_{i}, & 
        A_{-} = \displaystyle\sum_{i=1}^{N}\lambda_{i} q_{i},
        &
        K = \displaystyle\sum_{i=1}^{N}\left[q_{i}p_{i}
        -\displaystyle\frac{\lambda_{i}^{2}}{2}\right],
        \\
        B_{+} =\displaystyle\sum_{i=1}^{N} p^{2}_{i}, & 
        B_{-} =\displaystyle \sum_{i=1}^{N}q^{2}_{i}, & 
        M =\displaystyle\sum_{i=1}^{N} \lambda^{2}_{i},
    \end{array}
\label{eq:h6reaN}          
\end{equation}      
where we omit the realization symbol $\mathrm{D}(\ )$ for the sake of brevity.
Note that, formula~\eqref{eq:h6reaN} can also be written in in a vector notation as:
\begin{equation}  
    \begin{array}{lll}
        A_{+} =  (\vec \lambda \cdot \vec p), & 
        A_{-} = (\vec \lambda \cdot \vec q),
        &
        K = (\vec q \cdot \vec p) -\displaystyle\frac{\vec\lambda^{2}}{2},
        \\
        B_{+} =\vec p^2, & 
        B_{-} = \vec q^2, & 
        M =\vec \lambda^2,
    \end{array}
\end{equation}  
where $\vec q \coloneqq (q_1, \dots, q_N ) $, $\vec p \coloneqq (p_1, \dots,
p_N) $ and $\vl \coloneqq (\lambda_1, \dots, \lambda_N) $. 

In the realization~\eqref{eq:h6reaN} for $N=1,2$ the Casimir element~\eqref{eq:Cash6} becomes trivial: $\mathrm{D}(C) =0 $.
However, using higher-order coproducts~\eqref{higher-order-coproducts} for $m=2, \dots, N$, one can construct the following $2N-3$ functionally
independent Casimir functions, left and right respectively:
\begin{equation}
    C^{[m]} \coloneqq \Delta^{(m)}_L (C), 
    \qquad 
    C_{[m]} \coloneqq \Delta^{(m)}_R (C), 
\end{equation} 
which, after some calculations, can be explicitly written in the following form:
\begin{subequations}\label{left-and-right-casimirs} 
 	 \begin{align} \label{eq:left-casimirs} 
 	C^{[m]} = \sum_{1 \leq i < j < k}^m 
 	&\big( \lambda_i (p_j q_k - p_k q_j) 
 	+\lambda_j (p_k q_i - p_i q_k ) 
 	+\lambda_k (p_i q_j - p_j q_i)
 	\big)^2, 
 	\\  
 	C_{[m]} = \sum_{N-m+1 \leq i < j  < k }^N &	\big( \lambda_i (p_j q_k - p_k q_j) 
 	+\lambda_j (p_k q_i - p_i q_k ) 
 	+\lambda_k (p_i q_j - p_j q_i)
 	\big)^2.  \label{eq:right-casimirs} 
\end{align}
\end{subequations}
According to \Cref{thm:sumup}, all systems admitting coalgebra symmetry with
respect to the coalgebra $S\h_{6}$ will admit those invariants.  For this
reason, we call these invariants the \emph{universal
invariants}~\cite{BallesterosHerranz2007} for the systems admitting coalgebra
symmetry with respect to the $S\h_{6}$ coalgebra.
 
\begin{remark}\label{rmk-2N-3}
    We remark that for $m < N$, the invariants $C^{[m]} $ (resp. $C_{[m]} $)
    depend only on the variables $\{(q_1, p_1), \dots, (q_m, p_m)\}$ (resp.
    $\{(q_{N-m+1}, p_{N-m+1}), \dots, (q_N, p_N)\}$. Hence, these
    invariants are \textit{local}. Moreover, for $m=N$ the invariants
    $C^{[N]}$ and $ C_{[N]}$ are \textit{non-local} because they depend on
    all variables.  However, we have: 
    \begin{equation}
        C^{[N]} = C_{[N]} \mathop{=}^\eqref{eq:Cash6} C. 
    \end{equation}
    So, in the end, we have that through the coalgebra symmetry construction, applied to the $S\h_{6}$ algebra, we obtain $2N-5$ functional independent invariants.
 \end{remark}

\section{Classification results}
\label{sec:main-result}
The aim of this section is to prove a classification theorem for systems in a quasi-standard form admitting coalgebra symmetry with respect to the
$S\h_{6}$ algebra. This is the content of the following:

\begin{theorem}    \label{thm:h6} 
    A system in quasi-standard form admits coalgebra symmetry
    with respect to the Lie--Poisson coalgebra $\h_6$ if and
    only if
    \begin{equation}\label{eq:main-theorem-eq}
        \ell_k(\xi) = \xi,
        \quad
        V(\vq) = V(\vec \lambda \cdot \vec q , \vec q^2, \vl^2 ).
    \end{equation}
\end{theorem}
\begin{remark}  \label{rem:potabstr} 
    We observe that, \Cref{thm:h6} implies that the potential $V$ in the
    abstract form (\textit{i.e.}  in terms of the generators of $\h_6$
    algebra) \eqref{eq:main-theorem-eq} reads as 
    \begin{equation}
        \label{eq:corollary-potential} 
	V = V (A_-, B_-, M ). 
    \end{equation}
\end{remark}

We divide the proof of the main  \Cref{thm:h6} into several steps.  The first
one is the following lemma, which tells us that in fact the systems we are
looking for, are in standard form:
\begin{lemma}
    Let us assume that a system in quasi-standard from admits coalgebra
    symmetry with respect to the algebra $S\h_{6}$. Then, $\ell_k(\xi)=\xi$.
    \label{lem:ellk}
\end{lemma}

\begin{proof}
    Let us work in an $N$ degrees of freedom symplectic realization of
    $\mathfrak{h}_6$~\eqref{eq:h6reaN} with respect to the discrete canonical
    variables $( \vec q(t), \vec p(t) )$. Then the time evolution of the
    generator $B_{+}(t)$ under the system~\eqref{eq:dEL-quasistandard} is given
    by: 
    \begin{equation}\label{eq:Bp-evolution-lemma}
        B_{+}(t+h) = \sum_{i=1}^N p_i^2 (t+h) 
                   \mathop{=}^{ \eqref{eq:dEL-quasistandard-b} } \sum_{i=1}^N \left[ \ell_i' (\xi_i) \right]^2 q_i^2(t). 
    \end{equation}  
    Now, similarly to~\cite{GubLat_sl2}, we demand the independence of the
    right-hand side from the shifted variables $\vec q_k (t+h)$, \textit{i.e.}
    the derivative with respect to $q_l (t+h)$ can be put equal to zero: 
    \begin{equation}
        \ell_k' (\xi_k) \, \ell_k''(\xi_k) \, q_k^3(t) =0, 
        \qquad k=1, \dots, N,
    \end{equation}
    where $\xi_k \coloneqq q_k(t) q_k (t+h)$.
    Since $\ell_k' (\xi_k) \ne 0 $, we have 
    \begin{equation} \label{eq:ODE-lemma1}
        \ell_k'' (\xi_k) = 0. 
    \end{equation}
    The solution of \eqref{eq:ODE-lemma1}, including two constants, is
    \begin{equation}
        \ell_k (\xi_k) = c_k   \xi_k + d_k, \qquad k=1, \dots, N.
    \end{equation}
    We can safely put the constant $d_k$ to zero since it does not contribute to the
    equations of motion. Using the same argument as in the~\cite[Theorem
    3.1]{GubLat_sl2}, we can obtain that $c_k=c$ for all $k=1,\ldots,N$. Then, using the
    scaling 
    \begin{equation}
        V (\vec q) \mapsto \sqrt{c} V(\vec q), 
    \end{equation}
    we can assume without loss of generality $c=1$. Bringing everything
    together, we have $\ell_k (\xi) = \xi$, which proves the lemma.
\end{proof}

\begin{proposition}
    For generic $V$ the form of the associate system is: 
    \begin{subequations} \label{system}
    	\begin{align}
    	 K(t+h) & = -K(t) + \vec q(t) \cdot \nabla V - M(t),  \label{eq:Khs} 
    		\\ 
    		A_{+}(t+h) &=A_{-} (t),
            \label{eq:Aphs}
            \\ 
             A_{-}(t+h) &=\vec \lambda \cdot  \nabla V - A_{+}(t) 
             \label{eq:Amhs}
             \\ 
              B_{+}(t+h) &= B_{-} (t) 
            \label{eq:Bphs}
            \\ 
             B_{-}(t+h) &=\left( \nabla V \right)^2 
            - 2 \vec p(t)\cdot  \nabla V + B_{+} (t) 
            \label{eq:Bmhs}
            \\ 
            M(t+h) &= M(t). 
            \label{eq:Mhs}
    	\end{align}
        \label{eq:syss}
    \end{subequations}
    \label{prop:system}
\end{proposition}
\begin{proof}
    We insert the result of the Lemma~\ref{lem:ellk} back to the equation
    \eqref{eq:Bp-evolution-lemma}:
    \begin{equation}
        B_+(t+h) =  \sum_{i=1}^N q_i^2(t) 
                \mathop{=}  B_- (t).
    \end{equation}
    Then we can proceed with all other generators: 
    \begin{subequations}
        \begin{equation}
            A_{+}(t+h) = \sum_{i=1}^N \lambda_i p_i (t+h)
            = \sum_{i=1}^N \lambda_i \ell_i'(\xi_i) q_i(t)
            =  \sum_{i=1}^N \lambda_i q_i (t) = A_{-}(t), 
        \end{equation}
        \begin{align}
            B_{-}(t+h) =  \sum_{i=1}^N q_i^2(t+h) 
            = \sum_{i=1}^N \left[ \frac{\partial V (\vec q (t)) }{\partial q_i(t)}  - p_i (t) \right]^2
            = (\nabla V)^2 - 2 \vec p \cdot \nabla V + B_{+}(t). 
        \end{align}
    \end{subequations}
        By direct computations, in the same fashion one can obtain the
        equations of motion for all the generators, thus ending
        the proof of the proposition.
\end{proof}

\begin{proposition}
    The potential $V=V(\vq)$ satisfies the following overdetermined
    system of partial differential equations:
    \begin{subequations}
        \begin{align}
            [\vl^2\vq^2-(\vl\cdot\vq)^2](\grad V)^2=
            [\vl (\vq\cdot\grad V) - \vq (\vl\cdot\grad V)]^2  
            \label{eq:pdesysnlin}
            \\
            \left[ \vl^2\vq^2 - (\vl\cdot\vq)^2 \right] \frac{\partial V}{\partial q_j}
            = \left[ \vl^2 q_j - (\vl \cdot \vq) \lambda_j \right] (\vq \cdot \nabla V) 
            + \left[ \vq^2 \lambda_j - (\vl \cdot \vq) q_j  \right] (\vl \cdot \nabla V) 
            \label{eq:pdesyslin}
        \end{align}
        \label{eq:pdesys}
    \end{subequations}
    \label{prop:pdesys}
\end{proposition}
 
\begin{proof}   
    Now let us impose the condition (\ref{def-iii}) of the Definition~\ref{def-coalg-symmetry-discrete}: 
    \begin{equation}
        C(t+h) = C(t).
        \label{Casimir-interm}
    \end{equation}
    Substituting the equations~\eqref{system} to~\eqref{Casimir-interm}, we
    have, after some algebraic manipulations: 
    \begin{equation}
        \begin{aligned}
	    C(t+h) &= C(t) + \left[ \vec \lambda^2 \vec q^2 
  - (\vec \lambda \cdot \vec q)^2 \right] 
  (\nabla V)^2 
  - \left[  
  \vec \lambda  (\vec q \cdot \nabla V) 
  - \vec q (\vec \lambda \cdot \nabla V) 
  \right]^2 
        \\ 
         &+ 2 \sum_{j=1}^{N}  \bigg[ \left[  (\vec \lambda \cdot \vec q)^2
         - \vec q^2 \vec \lambda^2 \right] \frac{\partial V}{\partial q^j} 
       + \left[  \vec \lambda^2  q_j 
       - (\vec \lambda \cdot \vec q)  \lambda_j  
       \right]
         (\vec q \cdot \nabla V) 
       +\left[  \vec q^2 \lambda_j 
       - (\vec \lambda \cdot \vec q) q_j
       \right]
       (\vec \lambda \cdot \nabla V) 
       \bigg] \, p_j. 
        \end{aligned}
        \label{eq:Ctcond}
    \end{equation}
    Since $V$ does not depend on $\vec p$, in equation~\eqref{eq:Ctcond} we can
    take the coefficients with respect to the powers of $p_{j}$. For instance
    the constant term is:
    \begin{align} \label{p0-coeffts}
        \left[ \vec \lambda^2 \vec q^2 
  - (\vec \lambda \cdot \vec q)^2 \right] 
  (\nabla V)^2 
  - \left[  
  \vec \lambda  (\vec q \cdot \nabla V) 
  - \vec q (\vec \lambda \cdot \nabla V) 
  \right]^2  = 0, 
    \end{align}
    while the coefficients with respect to $p_j$ for $j=1,\dots,N$ are:
    \begin{align} \label{p1-coeffts}
    \left[  (\vec \lambda \cdot \vec q)^2
             - \vec q^2 \vec \lambda^2 \right] \frac{\partial V}{\partial q^j} 
           + \left[  \vec \lambda^2  q_j 
           - (\vec \lambda \cdot \vec q)  \lambda_j  
           \right]
             (\vec q \cdot \nabla V) 
           +\left[  \vec q^2 \lambda_j 
           - (\vec \lambda \cdot \vec q) q_j
           \right]
           (\vec \lambda \cdot \nabla V)  = 0.
    \end{align}
    That is, we obtained that the potential $V$ must satisfy the
    system~\eqref{eq:pdesys}, and conclude the proof of the proposition.
\end{proof}

Before considering the solution of the system~\eqref{eq:pdesys} in general,
we observe that for $N=3$ it is easy to prove that the potential must be of
the form $V=V(\vl\cdot\vq,\vq^2,\vl^2)$ through some simple geometric
considerations. Indeed, for $N=3$ note that we can use the triple vector product identity and Lagrange's identity~\cite{morse1953methods}:
\begin{equation}
    (\vec{a}\times\vec{b})^2 = \vec{a}^2\vec{b}^2-(\vec{a}\cdot\vec{b})^2,
    \quad
    \vec{a}\times\vec{b}\times\vec{c} = (\vec{a}\cdot\vec{c})\vec{b}-
    (\vec{a}\cdot\vec{b})\vec{c},
    \label{eq:valg}
\end{equation}
so that:
\begin{equation}
    \vl^2\vq^2-(\vl\cdot\vq)^2 = (\vl\times\vq)^2,
\end{equation}
and
\begin{equation}
    \vl (\vq\cdot\grad V) - \vq (\vl\cdot\grad V)=\vl\times\vq\times\grad V.
\end{equation}
Using Lagrange's identity once more:
\begin{equation}
    \left[\vl (\vq\cdot\grad V) - \vq (\vl\cdot\grad V)\right]^2 = 
    (\vl\times\vq)^2(\grad V)^2-[(\vl\times\vq)\cdot\grad V]^2,
    \label{eq:triplesquared}
\end{equation}
brings \cref{eq:pdesysnlin} into the simplified form:
\begin{equation}
    (\vl\times\vq)\cdot\grad V = 0.
\end{equation}
This means that the gradient of $V$ must be orthogonal to the vector
$\vl\times\vq$, i.e. $\grad V \in \Span(\vl,\vq)$. This implies that
$V=V(\vl\cdot\vq,\vq^2)$ or, adding the constant term $\vl^2$, $V=V(\vl\cdot\vq,\vq^2,\vl^2)$. Inserting this solution into the three equations
in formula \eqref{eq:pdesyslin}, we see that these are identically satisfied, thus ending the proof for $N=3$.

One can think that this kind of proof could be generalized to arbitrary $N$.
However, by direct inspection, it is possible to see the same strategy fails for $N>3$, since there is no direct analog of \cref{eq:triplesquared}. 
So, we will present a different proof, of which we don't have such a simple geometric interpretation. For the generic case, we proceed
by studying the linear condition given by the system~\eqref{eq:pdesyslin}, for which we prove the following proposition.

\begin{proposition}
    The system~\eqref{eq:pdesyslin} can be written as
    \begin{equation}
        \mathcal{M}_N(\vq,\vl)\grad V = \vec{0},
        \quad
        \mathcal{M}_N(\vq,\vl) = [(\vl\cdot\vq)^2-\vq^2\vl^2]\I_N+
        \mathcal{N}_N(\vq,\vl),
        \label{eq:sysmat}
    \end{equation}
    where
    \begin{equation}
        \mathcal{N}_N(\vq,\vl)
        =
        \left( q_i(\vl^2q_j-\vl\cdot\vq\lambda_j) +
        \lambda_i(\vq^2\lambda_j-\vl\cdot\vq q_j)\right)_{i,j=1}^{N}
    \end{equation}
    and the following inequality holds: 
    \begin{equation}
        N-2\leq\rank \mathcal{M}_N(\vq,\vl) \leq N.
    \end{equation}
    \label{prop:rank}
\end{proposition}

This result is proven using the following technical lemma:

\begin{lemma}
    Consider the $n\times n$ matrix:
    \begin{equation}
        M_{n}(\mu,\nu)
        =
        \begin{pmatrix}
            \mu & \nu & \cdots & \nu
            \\
            \nu & \ddots & \ddots & \vdots
            \\
            \vdots & \ddots & \mu & \nu
            \\
            \nu & \cdots & \nu & \mu
        \end{pmatrix}
        \label{eq:Mmunu}
    \end{equation}
    then:
    \begin{equation}
        \det M_n (\mu,\nu)
        =
        (\mu-\nu)^{n-1}[\mu+(n-1)\nu].
        \label{eq:sigman}
    \end{equation}
    \label{lemma:det}
\end{lemma}

\begin{proof}
    We prove this lemma via strong induction. Indeed, for
    $n=2$ by direct computation we have:
    \begin{equation}
        \det M_2 (\mu,\nu)
        =
        \det
        \begin{pmatrix}
            \mu & \nu
            \\
            \nu & \mu
        \end{pmatrix}
        =
        (\mu-\nu)(\mu+\nu)=(\mu-\nu)^{2-1}[\mu+(2-1)\nu],
        \label{eq:sigma2}
    \end{equation}
    that is~\cref{eq:sigman} for $n=2$.
    For $n>2$ consider the following decomposition of the matrix $M_n$:
    \begin{equation}
        M_{n}(\mu,\nu)
        =
        \begin{pmatrix}
            M_{n-2}(\mu,\nu) & \nu A
            \\
            \nu A^{T} & M_{2}(\mu,\nu)
        \end{pmatrix},
    \end{equation}
    where:
    \begin{equation}
        A = 
        \begin{pmatrix}
            1 & 1
            \\
            \vdots & \vdots
            \\
            1 & 1
        \end{pmatrix}.
        \label{eq:Amat}
    \end{equation}
    So, from the determinant of the block matrix formula we have:
    \begin{equation}
        \det M_{n}(\mu,\nu)
        =
        \det M_{2}(\mu,\nu)
        \det (M_{n-2}(\mu,\nu) - \nu^2 A M_{2}^{-1}(\mu,\nu)A^{T}).
        \label{eq:Mnexp}
    \end{equation}
    From a direct computation:
    \begin{equation}
        A M_{2}^{-1}(\mu,\nu)A^{T}
        =
        \frac{2}{\mu+\nu}
        \begin{pmatrix}
            1 & \cdots & 1
            \\
            \vdots & & \vdots
            \\
            1 & \cdots & 1
        \end{pmatrix}.
    \end{equation}
    So, inserting it in \eqref{eq:Mnexp} we have:
    \begin{equation}
        \det M_{n}(\mu,\nu)
        =
        \det M_{2}(\mu,\nu)
        \det M_{n-2}\left(\mu -\frac{2\nu^2}{\mu+\nu},\nu-\frac{2\nu^2}{\mu+\nu}\right).
    \end{equation}
    Using the inductive hypothesis and \cref{eq:sigma2} we have:
    \begin{equation}
        \begin{aligned}
        \det M_{n}(\mu,\nu)
        &= (\mu+\nu)
        (\mu-\nu)^{n-2}
        \left[\mu -\frac{2\nu^2}{\mu+\nu}
        +(n-3)\left(\nu-\frac{2\nu^2}{\mu+\nu}\right)\right]
        \\
        &=(\mu-\nu)^{n-1}
        \left[\mu +(n-1)\nu\right],
        \end{aligned}
    \end{equation}
    that is \cref{eq:sigman}. This concludes the proof of the lemma.
\end{proof}

\begin{proof}[Proof of \Cref{prop:rank}]
    Since the system \eqref{eq:pdesyslin} is linear, it always possible to write in matrix form as in \cref{eq:sysmat}.
    The shape of matrix $\mathcal{M}_N(\vq,\vl)$ follows from
    \cref{eq:pdesyslin} noting that all entries $\partial V/\partial q_i$
    are multiplied by the same factor, and using the properties of
    the scalar products.

    To estimate from below the rank of the matrix $\mathcal{M}_N(\vq,\vl)$
    we show that there exists a point in the space of parameters $\vl$ and 
    $\vq$ such that the matrix has \emph{exactly} rank $N-2$. Indeed, let us consider the points:
    \begin{equation}
        \vl_0 = (1,0,\ldots,0)^T,
        \quad
        \vq_0 = (1,\ldots,1)^T.
    \end{equation}
    Then we have
    \begin{subequations}
        \begin{gather}
            (\vl_0\cdot\vq_0)^2-\vq_0^2\vl_0^2 = -(N-1),
            \\
            \mathcal{N}_N(\vq_0,\vl_0)_{i,j} = 1 - \lambda_{0,i}-
            \lambda_{0,j}
            +N\lambda_{0,i}\lambda_{0,j}
            =
            \begin{cases}
                N-1 & i=j=1,
                \\
                1 & i,j\neq 1,
                \\
                0 & \text{otherwise},
            \end{cases}
        \end{gather}
    \end{subequations}
    so that:
    \begin{equation}
        \mathcal{M}_N(\vq_0,\vl_0) =
         \begin{pmatrix}
            0 & 0 & 0 & \cdots & 0
            \\
            0 &-(N-2) & 1 & \cdots & 1
            \\
            0 &1 & \ddots & \ddots & \vdots
            \\
            \vdots &\vdots & \ddots & -(N-2) & 1
            \\
            0 &1 & \cdots & 1 & -(N-2)
        \end{pmatrix}
    \end{equation}
    Clearly $\rank\mathcal{M}_N(\vq_0,\vl_0)<N$ since $\det\mathcal{M}_N(\vq_0,\vl_0)=0$. Consider now the submatrix $\mathcal{M}_N(\vq_0,\vl_0)_{[1,1]}$ obtained
    removing the first row and column of $\mathcal{M}_N(\vq_0,\vl_0)$.
    Then we can compute its determinant using \Cref{lemma:det}:
    \begin{equation}
        \det \mathcal{M}_N(\vq_0,\vl_0)_{[1,1]}
        =
        \det M_{N-1}(-(N-2),1) = 0.
    \end{equation}
    All the other $(N-1)\times(N-1)$ will have determinant zero because
    they will necessarily have one row and/or one column full of zeroes,
    so $\rank\mathcal{M}_N(\vq_0,\vl_0)<N-1$.
    Consider now the submatrix $\mathcal{M}_N(\vq_0,\vl_0)_{[\{1,2\},\{1,2\}]}$ obtained
    removing the first and second row and column of $\mathcal{M}_N(\vq_0,\vl_0)$.
    Then we can compute again its determinant using \Cref{lemma:det}:
    \begin{equation}
        \det \mathcal{M}_N(\vq_0,\vl_0)_{[\{1,2\},\{1,2\}]}
        =
        \det M_{N-2}(-(N-2),1) = (-1)^N (N-1)^{N-3} \neq 0.
    \end{equation}
    This gives us $\rank\mathcal{M}_N(\vq_0,\vl_0)=N-2$, and hence
    we obtained the desired lower bound for $\mathcal{M}_N(\vq,\vl)$. 
    This ends the proof of the proposition.
\end{proof}

With the following proposition, we are able to characterize the non-trivial
solution of the system~\eqref{eq:pdesys}, thus giving the final step to
prove \Cref{thm:h6}.

\begin{proposition}
    The only non-trivial solution of the system~\eqref{eq:pdesys}
    is $V=V(\vl\cdot\vq,\vq^2,\vl^2)$.
    \label{prop:solsys}
\end{proposition}

\begin{proof}
    First, let us check the proposition for the equation \eqref{eq:pdesyslin}. Since this equation is linear, it is enough to check it for the test functions $v_a = (\vec \lambda \cdot \vec q) $ and $v_b = \vec q^2 $. Hence, we get 
    \begin{align}
         \nabla v_a = \vec \lambda, 
        \qquad 
        & \nabla v_b = 2 \vec q, 
        \\
        \frac{\partial v_a}{\partial q_j} = \lambda_j, \qquad 
        & \frac{\partial v_b}{\partial q_j}=
        2q_j. 
    \end{align}
    Now it is easy to see that \eqref{eq:pdesyslin} becomes identity in case of $V= v_a$: 
    \begin{align}
       -  \left[ \vl^2\vq^2 - (\vl\cdot\vq)^2 \right] \lambda_j 
            + \left[ \vl^2 q_j - (\vl \cdot \vq) \lambda_j \right] (  \vec \lambda \cdot \vq ) 
            + \left[ \vq^2 \lambda_j - (\vl \cdot \vq) q_j  \right] 
            \vl^2  =0,
    \end{align}
    and $V= v_b$: 
\begin{align}
    - \left[ \vl^2\vq^2 - (\vl\cdot\vq)^2 \right] \cdot  2q_j 
            + \left[ \vl^2 q_j - (\vl \cdot \vq) \lambda_j \right] 
            2 \vq^2  
            + \left[ \vq^2 \lambda_j - (\vl \cdot \vq) q_j  \right] (\vl \cdot \vq ) =0. 
\end{align}
The equation \eqref{eq:pdesysnlin} is non-linear, hence one has to check it for the full expression 
\begin{align}
      \nabla V \left( (\vec \lambda \cdot \vec q),
    \, 
    \vec q^2, \, \vec \lambda^2
    \right)
    &= \frac{\partial V}{\partial ( \vec \lambda \cdot \vec q ) } 
    \frac{\partial ( \vec \lambda \cdot \vec q ) }{ \partial \vec q}
    + \frac{\partial V}{\partial (\vec q^2) }
    \frac{\partial \vec q^2}{\partial \vec q} 
   = \vec \lambda \frac{\partial V}{\partial A_-} + 2 \vec q \frac{\partial V}{\partial B_-}. 
\end{align}
Here we used \eqref{eq:h6reaN}. Hence we have 
\begin{align}
    (\nabla V)^2 & = \vec \lambda^2 \left(\frac{\partial V}{\partial A_-} \right)^2 
    +4 (\vec \lambda \cdot \vec q) 
    \frac{\partial V}{\partial A_-}
    \frac{\partial V}{\partial B_-}
    +4 \vec q^2 \left( \frac{\partial V}{\partial B_-} \right)^2 
\notag 
\\
    &= M \left(\frac{\partial V}{\partial A_-} \right)^2 
    +4 A_- 
    \frac{\partial V}{\partial A_-}
    \frac{\partial V}{\partial B_-}
    +4 B_- \left( \frac{\partial V}{\partial B_-} \right)^2, 
    \label{eq:solutions-nonlinear1}
\end{align}
\begin{equation}
 (\vec q \cdot \nabla V) = 
 (\vec \lambda \cdot \vec q) \frac{\partial V}{\partial A_-} + 2 \vec q^2 \frac{\partial V}{\partial B_-}
 =
 A_- \frac{\partial V}{\partial A_-} + 2 B_- \frac{\partial V}{\partial B_-},
 \label{eq:solutions-nonlinear2}
\end{equation}

\begin{equation}
    (\vl \cdot \nabla V) = \vl^2 \frac{\partial V}{\partial A_-}
    + 2 ( \vl \cdot \vq) \frac{\partial V}{\partial B_-} 
    = 
    M \frac{\partial V}{\partial A_-} 
    + 2 A_- \frac{\partial V}{\partial B_-}.
    \label{eq:solutions-nonlinear3}
\end{equation}
    Now, substituting \eqref{eq:solutions-nonlinear1},
    \eqref{eq:solutions-nonlinear2} and \eqref{eq:solutions-nonlinear3} to
    \eqref{eq:pdesysnlin} and performing some algebra, we again obtain an identity.
\end{proof}

The final step of the proof of \Cref{thm:h6} follows from the following
proposition, that we will also be used extensively in
\cref{sec:integrability-props} when looking for polynomial invariants is the
following:
\begin{proposition}
    Let us assume that a system in quasi-standard form admits coalgebra
    symmetry with respect to the Lie--Poisson algebra $\h_{6}$. Then, the associated
    dynamical system on the generators of the algebra has the following form:
    \begin{subequations}
        \begin{align}
            K(t+h) &=  
            \label{eq:Khf}
            -K(t) + A_{-}(t) \frac{\partial V}{\partial A_{-}} 
            +2B_{-}(t) \frac{\partial V}{\partial B_{-}}  - M(t), 
            \\
            A_{+}(t+h) &=  A_{-}(t), 
            \label{eq:Aphf}
            \\
            A_{-}(t+h) &= M (t) \frac{\partial V}{\partial A_{-}} 
            +2 A_{-} (t) \frac{\partial V}{\partial B_{-}}
            - A_{+} (t),  
            \label{eq:Amhf}
            \\
            B_{+}(t+h) &= B_{-} (t), 
            \label{eq:Bphf}
            \\
            B_{-}(t+h) &= M (t) \left( \frac{\partial V}{\partial A_{-}}   \right)^2 
            +4 A_{-}(t) \frac{\partial V}{\partial A_{-}}  \frac{\partial V}{\partial B_{-}}
            -4 A_{-}(t) \frac{\partial V}{\partial B_{-}}
            +4 B_{-}(t) \left(\frac{\partial V}{\partial B_{-}} \right)^2 
            \label{eq:Bmhf}
            \\ \notag 
            &\phantom{=} -2 A_{+}(t) \frac{\partial V}{\partial A_{-}} 
            +B_{+} (t)
            \\
            M(t+h) &= M(t). 
            \label{eq:Mhf}
        \end{align}
        \label{eq:sysfin}
    \end{subequations}
    \label{cor:system}
\end{proposition}

\begin{proof}
    Direct computation using the solution found in \Cref{prop:solsys}
    into the system~\eqref{system} from \Cref{prop:system}.
\end{proof}

\begin{proof}[proof of \Cref{thm:h6}]
    The proof follows from \Cref{cor:system}, noting that the discrete-time
    evolution of the generators of the coalgebra~\eqref{eq:sysfin} is closed in
    $\h_{6}$ and satisfy the Poisson map condition. Finally, the condition on
    the Casimir elements follows from~\eqref{eq:Mhf} and again by
    \Cref{prop:solsys}.
\end{proof}

We conclude this section, with the following corollary, which will be used
extensively in \cref{sec:sys}:

\begin{corollary}
    The Lagrangian of the system in quasi-standard form, possessing coalgebra symmetry with respect to the Lie--Poisson coalgebra $\mathfrak{h}_6$ can be written as follows: 
    \begin{equation}
        L = \sum_{k=1}^N  q_k(t+h) q_k (t) -  V(\vec \lambda \cdot \vec q , \vec q^2, \vl^2 ), 
        \label{eq:Lgenh6}
    \end{equation}
    whose discrete Euler--Lagrange equations are given by:
    \begin{equation}  
     q_k(t+h)
    + q_k(t-h) 
    = 2\frac{\partial  V(\vec \lambda \cdot \vec q , \vec q^2, \vl^2 )}{\partial \vec{q} (t)}q_{k}(t)
+ \frac{\partial  V(\vec \lambda \cdot \vec q , \vec q^2, \vl^2 )}{\partial (\vec{\lambda}\cdot \vec{q} (t))}\lambda_{k},
        \quad
        k=1,\dots,N.
        \label{eq:dELgenh6}
    \end{equation}
    Finally, the canonical form of the equation of motion is the following one:
    \begin{subequations}        \label{eq:cangenh6} 
        \begin{align}
            q_{k}(t+h) + p_{k}(t) &= 2\frac{\partial  V(\vec \lambda \cdot \vec q , \vec q^2, \vl^2 )}{\partial \vec{q} (t)}q_{k}(t)
    + \frac{\partial  V(\vec \lambda \cdot \vec q , \vec q^2, \vl^2 )}{\partial \left(\vec{\lambda}\cdot \vec{q} (t)\right)}\lambda_{k},
            \label{eq:qkgenh6}
            \\
            p_{k}(t+h) &= q_{k}(t).
            \label{eq:pkgenh6}
        \end{align}
    \end{subequations}
    \label{cor:genh6}
\end{corollary}

\section{Search for polynomial invariants}
\label{sec:integrability-props}

The coalgebra symmetry approach allows us to obtain $2N-5$ functionally
independent Casimir functions~\eqref{left-and-right-casimirs} of which $N-2$
commute. Such number is not enough to ensure Liouville integrability even in
the continuous setting, see~\cite{blascosanzIntegrabilidadSistemasNo2009},
since upon addition of the Hamiltonian one obtains only $N-1$ commuting
invariants. Moreover, since in the discrete setting we do not have a
straightforward equivalent of a Hamiltonian, see for instance the discussion
in~\cite{Gubbiotti_EtAl2023.Coalgebra_symmetry_discrete_systems} and
\Cref{rmk-2N-3}, we are actually missing not just one, but two invariants. 

In this section, we look for admissible potentials $V$ possessing  polynomial
invariants in the non-trivial generators of the two-photon algebra
$\mathfrak{h}_{6}$. We will look for the existence of invariants for the system
on the generators of the coalgebra~\eqref{eq:sysfin}, since they will be
invariants for the underlying symplectic map~\eqref{eq:cangenh6}, see
also~\cite[\S 7]{Gubbiotti_EtAl2023.Coalgebra_symmetry_discrete_systems}.
Motivated by the known examples,
see~\cite{Gubbiotti_EtAl2023.Coalgebra_symmetry_discrete_systems}, we look for
invariants either linear or quadratic in the realization~\eqref{eq:h6reaN}.
Moreover, we have the consistency condition on the potential: 
\begin{equation}  \label{eq:conscondpot}
    \frac{\partial V}{\partial A_- } \ne 0.
\end{equation}
Indeed, if a potential is such that $\partial V / \partial A_{-}=0$,
then $V=V(B_{-},M)$. That is, $V$ is a potential expressible in terms of the
generators of Lie subalgebra $\mathfrak{gl}_2(\R)\subset \mathfrak{h}_{6}$,
which is a trivial central extension of $\mathfrak{sl}_2(\R)$, the case was already
studied in~\cite{GubLat_sl2}.

An invariant for the system~\eqref{eq:sysfin} satisfies the
condition~\eqref{invariants-general-def-t+h}, where in the left hand side we
substitute the explicit time evolution. This yields a mixed system of algebraic
equations, linear PDEs, and nonlinear PDEs, whose complexity increases greatly
with the degree of the polynomial. In our discussion, we separate solutions
that specify the pair potential-invariant(s) uniquely and solutions that do not
specify completely the potentials, that is where the potential depends on
arbitrary functions of $A_-$, $B_{-}$ and/or their functional combinations. We
will address the first kind of solution of the problem as \emph{non-singular
solutions}, and we will denote the second kind of as \emph{singular
solutions}. We underline that we also consider the possibility of having more
than one invariant of the same kind associated with a potential.

\subsection{Linear invariants} \label{sect:linear-invariants}

Let us first consider the case of linear invariant 
\begin{equation}\label{eq:invariant-linear}
    I_1 = \varkappa(M) K + \alpha_+(M) A_+ + \alpha_-(M) A_- 
    + \beta_+(M) B_+ + \beta_-(M) B_-,
\end{equation}
and set $ V  = V_1 (A_-, B_-, M )$ in
\eqref{eq:corollary-potential} and \eqref{system}. From now for the sake of simplicity, we will omit the explicit dependence of the arbitrary
constants on the generator $M=\vec{\lambda}^{2}$.
Then, the following proposition holds: 
\begin{proposition} \label{prop:linear}
    If the system~\eqref{system} admits a linear polynomial invariant, then
    there exists a unique non-singular solution given by:
   \begin{subequations}
   	 \begin{align} 
		\label{linear-potential-solution-united} 
                V_{1} (A_-, B_-, M) &= 
                    - \alpha_+ \, A_- - \frac{\varkappa}{2} B_-,
                    \\ 
                     \label{linear-invariants-two}
                I_{1} & =\varkappa \, K + \alpha_+ \, (A_+ + A_- ) 
	            + B_+ + B_-;
   	 \end{align}
	\label{eq:V1I1together}%
   \end{subequations}
   and a unique singular solution, given by:
   \begin{subequations}
       \begin{align}
            V^{(s)}_{1} &= B_{-} + F(M, \, M B_{-}-A_{-}^{2})                  
           \label{eq:V1s}; 
           \\ 
           I_{1}^{(s)} &=A_{+}-A_{-},
           \label{eq:I1s}
            \end{align}
       \label{eq:linsing}
   \end{subequations}
   where $F=F(X,Y)$ is an arbitrary function of two arguments.
\end{proposition}

\begin{proof}
    Let us take the most generic invariant linear in the generators of
    $\mathfrak{h}_{6}$ with $\{\varkappa, \alpha_+, \alpha_-, \beta_+,
    \beta_- \}  $ being undetermined functions, depending on $M$, not all
    identically zero. Now let us impose that~\eqref{eq:invariant-linear} is
    an invariant for the system~\eqref{system}, \emph{i.e.}\ $I_1(t+h) =
    I_1(t)$ on the solutions of~\eqref{system}. This yields the following
    condition:
    \begin{align}
		& \ 4\beta_- B_- \left(\frac{\partial V}{\partial B_-}\right)^{2}
		+\beta_- M \left(\frac{\partial V}{\partial A_-}\right)^{2} 
		+ 4 \beta_- A_- \frac{\partial V}{\partial A_-} 
	\frac{\partial V_1}{\partial B_-} 
	\notag \\ 
	 &+ 	2\left(  \alpha_- A_- + \varkappa B_- -\beta_- M- 2 \beta_- K  \right) \frac{\partial V_1}{\partial B_- }  
	+( \varkappa A_-+ \alpha_- M-2\beta_- A_+ ) \frac{\partial V_1 }{\partial A_-} 
	\notag \\ 
	&-2 \varkappa K  - \left( \alpha_- +\alpha_+\right) A_+
	+\left(\alpha_+ -\alpha_-  \right) A_-
	+\beta_-   B_+ + \beta_+ B_- 
	- \varkappa M=0. 
    \end{align}
    Since $V$ does not depend on $K, \, A_+, \, B_+$, we can take coefficients with respect to these generators. So, we obtain the following algebro-differential
    system: 
    \begin{subequations}  \label{eq:systgen10} 
    	\begin{align}  \label{eq:systgen10a}
    		 \beta_- &= \beta_+,
	   		 \\  \label{eq:systgen10b}
            2 \beta_{-}\frac{\partial V(A_-, B_-, M) }{\partial A_-} + \alpha_+ + \alpha_- &=0,
          \\  \label{eq:systgen10c}
            2 \beta_{-}\frac{\partial V(A_-, B_-, M) }{\partial B_-} + \varkappa &=0,
              	\end{align}
              \begin{equation}
              	  \begin{aligned}
 \beta_- M \left( \frac{\partial V}{\partial A_-}  \right)^2 
		&
		+ 4 \beta_- B_- \left( \frac{\partial V_1 }{\partial B_-} \right)^2  
		+ 4 \beta_- A_- \frac{\partial V}{\partial A_-} \frac{\partial V}{\partial B_-} 
                +2 (   \alpha_- A_- +\varkappa B_- - \beta_- M) 
		\frac{\partial V}{\partial B_-}
                \\
		&+ (\varkappa A_- + \alpha_- M ) \frac{\partial V}{\partial A_-}
		+ (\alpha_+ - \alpha_-) A_-
		 - \varkappa M + (\beta_+ - \beta_- ) B_- =0.
            \end{aligned} 
            \label{eq:systgen11}
              \end{equation}
               \label{eq:sys1degcoeff}
    \end{subequations}
    
    From the three relations~\eqref{eq:systgen10a}-\eqref{eq:systgen10c} we
    have that  $\beta_+ \neq 0$ makes the two linear partial differential
    equations non-singular.  So, the non-singular solution of the
    system~\eqref{eq:sys1degcoeff} is obtained by solving iteratively the two
    linear PDEs in~\eqref{eq:systgen10b}-\eqref{eq:systgen10c} and use the
    nonlinear one~\eqref{eq:systgen11} as a compatibility condition.  That is,
    from~\eqref{eq:systgen10a}-\eqref{eq:systgen10c} we obtain the following
    expression for the potential:
    \begin{equation}
	V (A_-, B_-, M) = f(M) - \frac{(\alpha_+ + \alpha_-) }{2 \beta_+ } 
	A_- 
        - \frac{\varkappa}{2 \beta_+ } B_-, 
        \label{eq:V10}
    \end{equation}
    where $f(M) $ is an arbitrary smooth function of its argument.
    Inserting~\eqref{eq:V10} in~\eqref{eq:systgen11}, we obtain two different solutions
    for the functions $\alpha_{+}$, $\alpha_{-}$, and $\varkappa$ as follows:
    \begin{subequations}
        \begin{align}
	    \alpha_- = \alpha_+, \quad
            &\alpha_{+}=\text{arbitrary},
            \quad
            \varkappa =\text{arbitrary}, 
            \quad 
            \beta_+  =\text{arbitrary}; 
            \label{eq:deg1sol1}
            \\
	    \alpha_- = - \alpha_+, 
            \quad
          & \alpha_{+}=\text{arbitrary},
	    \quad 
	    \varkappa = -2 \beta_+, 
	    \quad  
	    \beta_+  =\text{arbitrary}.  
            \label{eq:deg1sol2}
        \end{align}
    \end{subequations}
    Now, we note that putting~\eqref{eq:deg1sol2} into~\eqref{eq:V10}
    the term in $A_{-}$ vanishes, \emph{i.e.} such a solution does not satify
    the consistency condition~\eqref{eq:conscondpot}.
    So, only the first solution~\eqref{eq:deg1sol1} is acceptable: 
     \begin{equation}
 	V (A_-, B_-, M) = f (M) 
 	- \frac{\alpha_+}{\beta_+} A_- 
 	- \frac{\varkappa}{2 \beta_+ } B_-.
    \end{equation}
    Noting that the arbitrary function $f$ does not enter the equations of
    motion, we can safely put it to zero, and we obtain the first part of the
    statement.  Moreover, since we assumed $\beta_+ \ne 0 $, we can remove such
    parameter by rescaling the coefficients as  
    \begin{equation}
    	         \alpha_+ \mapsto \alpha_+ \beta_+,
    	         	\qquad  
    	         \varkappa \mapsto \varkappa \beta_+. 
    	     \end{equation} 
  So, we obtain
    the form of the potential and the invariant presented in
    equations~\eqref{eq:V1I1together}.

    Let us now consider the singular case, \emph{i.e.} $\beta_{+}\equiv0$
    in~\eqref{eq:sys1degcoeff}. Then, inserting $\beta_{+}=0$ into the
    algebro-differential system~\eqref{eq:systgen10}, and eliminating
    the already satisfied equation~\eqref{eq:systgen10a}, we obtain
    two constraints on the parameters $\alpha_{-}=-\alpha_{+}$ and $\varkappa=0$,
    and a single \emph{linear} PDE from~\eqref{eq:systgen11}:
    \begin{equation}
            \alpha_{+}\left[2 A_-
		\frac{\partial V}{\partial B_-}
                + M  \frac{\partial V}{\partial A_-}
            - 2 A_-\right]=0.
            \label{eq:systgen11sing}
    \end{equation}
    Then, noting that $\alpha_{+}=0$ yields a trivial solution, we are
    left to satisfy equation~\eqref{eq:systgen11sing} for $V$. Solving it,
    we obtain~\eqref{eq:V1s} concluding the proof of the proposition.
\end{proof}

\subsection{Invariants quadratic in the physical variables}
\label{sect-quadratic-in-realiz}

Following what we said in the introduction of this section, in this subsection
we characterize potentials, admitting invariants of degree $2$ in the
realization \eqref{eq:h6reaN}, \textit{i.e.} the \emph{physical variable}.
Since:
\begin{equation}
    \deg_{(\vec{q},\vec{p})} A_{\pm} = 1,
    \quad
    \deg_{(\vec{q},\vec{p})} B_{\pm} = 
    \deg_{(\vec{q},\vec{p})} K = 2,
    \label{eq:deg_gen}
\end{equation}
we consider invariants of the following form:     
 \begin{equation}\label{eq:I2} 
   I_2 = 
       \beta_+ B_+ + \beta_- B_- + \varkappa K + \alpha_\pm \, A_+ A_- 
       + \alpha_{1+} A_+ 
       + \alpha_{1-} A_-
       + \alpha_{2+} A_+^2
       + \alpha_{2-} A_-^2, 
\end{equation}
where again, for the sake of brevity, we omit the explicit dependence of the
coefficients on $M$.

\begin{proposition}\label{prop:quadr-phys}
    Let us assume that the system~\eqref{system}  admits polynomial invariants, quadratic in
    the physical variables $(\vec q, \vec p) $. Then there exist: 
 \begin{enumerate}
 	\item  two non-singular solutions, given by:
\begin{subequations}\label{eq:quadratic-qp-pair1}
	\begin{align}
			V_{2,I} & =
		\frac{1 }{2}\left[ 
			 \frac{\left(\kappa  +2  \right)}{ M} A_-^2 
			- \kappa B_-
			\right], 
		\label{eq:quadratic-qp-potential-2}
		\\ 
                I_{2,I} & = 
			-  \left( 
			1
			+ \frac{\kappa }{2}  \right)
			\frac{\left(A_+^2 + A_-^2 \right)}{M}
			+ \kappa \, K 
                        + B_+ + B_-; 
                \label{eq:quadratic-qp-invariants1}
	\end{align}
\end{subequations}
and 
\begin{subequations}\label{eq:quadratic-qp-pair2}
	\begin{align}
            V_{2,II}   &=-  \frac{\eta}{2} A_-^2 - \zeta A_- 
                - \frac{\kappa}{2}  B_-, 
	    \label{eq:quadratic-qp-potential} 
	    \\ 
	    I_{2,IIa} & =  
                \kappa (M K - A_+ A_-)
		 -  (
		 A_-^{2} 
		 +A_+^{2} )  
		 + (B_-   
		 + B_+    ) M, 
                 \label{eq:I2IIa}
		 \\
		 I_{2,IIb} &= 
                 (\eta M + \kappa) A_{+} A_- 
		 + \zeta M (A_+ + A_- )
		 + A_+^2 + A_-^2,  
                 \label{eq:I2IIb}
	    \end{align}
	\end{subequations} 
        where the new parameters $\zeta$, $\eta$, and $\kappa$ are in
        one-to-one correspondence with the old ones; 
    \item two singular solutions, where the potential depends on arbitrary
        functions of their arguments, given by:   
    \begin{subequations}
        \begin{align}\label{eq:V2Is}
        V_{2,I}^{(s)} &=  - \frac{\alpha}{2 } B_- + G (M, \, A_-),
                \quad \frac{\partial^{3} G}{\partial Y^{3}}(X,Y) \neq 0,
        	\\
        	\label{eq:I2Is}
        	 I_{2,I}^{(s)}&= 
	  A_+ A_-
	+\frac{1}{\alpha} (A_-^{2} + A_+^{2}
	-   B_- M 
	-B_+ M ) 
	-   K M; 
        \end{align}
    \end{subequations}
   and 
    \begin{subequations}
        \begin{align}\label{eq:V2IIs}
            V_{2,II}^{(s)} &= 
            - \frac{\alpha}{2 } B_-  
            - \frac{\alpha_{+}}{  M} A_- + F \left(M, MB_{-}- A_-^2 \right),
                \quad 
                \frac{\partial^{2} F}{\partial Y^{2}}(X,Y) \neq 0, 
                \\  \label{eq:I2IIs}
                	I^{(s)}_{2,II} & = 
                	A_+^2  
                	 + A_-^2 
                         +\alpha A_{+}A_{-}
                         + \alpha_{+} ( A_+ +  A_-),
   		\end{align}
   \end{subequations}
    where the new parameters $\alpha $ and $\alpha_+ $ are in
    one-to-one correspondence with the old ones. 
 \end{enumerate}   
\end{proposition}

\begin{remark}\label{rmk:I1s}
    The potential $V_{2, I} $~\eqref{eq:quadratic-qp-potential-2} is a
    particular case of the singular potential $V_{1}^{(s)}$~\eqref{eq:V1s} for
    $F(X,Y)=\varkappa Y/2\beta_{+} $. This implies that it admits also the
    invariant $I_{1}^{(s)}=A_{+}-A_{-}$ from~\eqref{eq:I1s}.
\end{remark}
 
\begin{proof}
    Let us now assume the existence of a generic polynomial invariant
    \eqref{eq:I2}, at most quadratic in physical variables $(\vec q, \vec p)$.
    We impose that~\eqref{eq:I2} is an invariant for the system~\eqref{system},
    \emph{i.e.}\ $I_2(t+h) = I_2(t)$ on the solutions of~\eqref{system}.
    Expanding this and taking the coefficients similarly as in the proof of the
    proposition~\ref{prop:linear}, we arrive at the following
    algebro-differential system: 
    \begin{subequations}\label{eq:quadratic-system} 
   	 \begin{align}\label{eq:quadratic-system-alg-beta} 
    	 \beta_+ &= \beta_-,
    	 \\
    	 \alpha_{2+} & = \alpha_{2-},\label{eq:quadratic-system-alg-alpha2}
    	 \\ 
    	 \label{eq:quadratic-system1} 
    	 	2( \beta_-+ \alpha_{2-} M ) \frac{\partial  V(A_-, B_-, M)}{\partial A_-} 
    	 + 4 \alpha_{2-} A_- \frac{\partial  V(A_-, B_-, M)}{\partial B_-} 
         + 2 \alpha_\pm A_- + \alpha_- + \alpha_+ &=0,  
     		\\ 
     		 \label{eq:quadratic-system2} 
     		  2 \beta_- \frac{\partial  V(A_-, B_-, M)}{\partial B_-}   + \varkappa &=0. 
     \end{align} 
     \begin{equation}\label{eq:quadratic-system3}
     	\begin{aligned}
     		  M \left( \alpha_{2-} M + \beta_- \right)
     		\left( \frac{\partial V}{\partial A_-} \right)^2
     	+4(	 \alpha_{2-} A_-^2 + \beta_- B_-  )
     	 \left( \frac{\partial V}{\partial B_-} \right)^2
          		+ 4 ( \alpha_{2-}M+\beta_- ) A_-
     		 \frac{\partial V}{\partial A_-} 
     		 \frac{\partial V}{\partial B_-} &
     		 \\ 
      + (\alpha_\pm M A_- + \varkappa A_- + \alpha_{1-}M  )
      \frac{\partial V }{\partial A_-}
      + 2 (\alpha_\pm A_-^2 + \alpha_{1-} A_- + \varkappa B_- - \beta_- M ) \frac{\partial V}{\partial B_-} &
      \\ 
      + (\alpha_{2+} - \alpha_{2-} ) A_-^2  
      + (\alpha_{1+} - \alpha_{1-} ) A_- 
      + (\beta_+ - \beta_- ) B_- 
      - \varkappa M &= 0.       
     	\end{aligned}
     \end{equation} 
        \end{subequations}

    From the four relations~\eqref{eq:quadratic-system-alg-beta}-\eqref{eq:quadratic-system2} we
    have that the non-singularity conditions are:
    \begin{equation}
          \beta_+ \neq 0,
          \qquad\text{and}\qquad
          \alpha_{2+}M+\beta_{+}\neq0.
    \end{equation}
    Then, we can proceed to solve the system~\eqref{eq:quadratic-system} as we
    did in \Cref{prop:linear}. For instance, we start by solving the equalities
    \eqref{eq:quadratic-system-alg-beta}-\eqref{eq:quadratic-system-alg-alpha2}
    and putting them into the two linear PDEs
    \eqref{eq:quadratic-system1}-\eqref{eq:quadratic-system2}. So, in
    the non-singular case this yields the
    solution:
    \begin{equation}\label{eq:general-soln-for-quadratic-potential}
	V (A_-, B_-, M) = f(M) 
	- \frac{1}{2} \left[ 
	 \frac{(\alpha_\pm \beta_+ - \alpha_{2+} \varkappa) }{\beta_+ (\alpha_{2+} M+\beta_+)}A_-^2
	+ \frac{(\alpha_{1+} + \alpha_{1-} )}{ ( \alpha_{2+} M+\beta_+ )} A_- 
    + \frac{\varkappa}{ \beta_+ } B_-\right].
    \end{equation}
    By substituting \eqref{eq:general-soln-for-quadratic-potential} into
    \eqref{eq:quadratic-system3}, we obtain two different sets of solutions
    with respect to the coefficients $\{ \alpha_{1+}, \,\alpha_{1-} \,
    \alpha_{2+},  \,  \alpha_\pm, \,  \varkappa, \, \beta_+ \} $ and obtain
    the two following potentials:  
    \begin{subequations}\label{eq:quadratic:solution}
            \begin{align} \label{eq:quadratic:solutionI} 
                    V_{2,I} (A_-, B_-, M) & \coloneqq f(M) 
                            + \frac{\left(\varkappa  +2 \beta_+  \right)}{2 \beta_+ \, M} A_-^2 
                            - \frac{\varkappa  }{2 \beta_+} B_-, 
                    \\ 
                    V_{2,II} (A_-, B_-, M)  & \coloneqq f  (M)
                    - \frac{(\alpha_\pm  \,  \beta_+   - \alpha_{2+}  \varkappa    )}{2 \beta_+ ( \alpha_{2+} \, M+ \beta_+  )} A_-^2 
                    - \frac{\alpha_{1+} }{ \alpha_{2+} \, M+ \beta_+ } A_- 
                    -  \frac{\varkappa}{2\beta_+ } B_-.
                    \label{eq:quadratic:solutionII}
            \end{align}	
    \end{subequations}
    We observe that both potentials are admissible because they satisfy the
    consistency condition~\eqref{eq:conscondpot}. Moreover, as we did earlier,
    we can remove the function $f$, as it does not affect the equations of
    motion. So, from~\eqref{eq:quadratic:solutionI} we obtain the potential
    \eqref{eq:quadratic-qp-potential-2}. Substituting the associated value of
    the parameters yields the invariant~\eqref{eq:quadratic-qp-invariants1}.
    Now, to simplify the expression of the
    potential~\eqref{eq:quadratic:solutionII}, we introduce some new
    coefficients as follows: 
    \begin{equation}\label{eq:transformation-parameters}
            \eta (M) \coloneqq  \frac{\alpha_\pm - \alpha_{2+} \varkappa}{\beta_+ (\alpha_{2+} M + \beta_+)},
            \qquad 
            \zeta(M) \coloneqq \frac{\alpha_{1+}}{\alpha_{2+}M + \beta_+}, 
            \qquad 
            \kappa (M) \coloneqq \frac{\varkappa}{\beta}. 
    \end{equation}
    Note that, this is justified because the
    transformation~\eqref{eq:transformation-parameters} allows us to remove
    uniquely three parameters, say $\varkappa$, $\alpha_{1,+}$, and
    $\alpha_{2+}$, and it is not defined only for the singular cases.  In this
    new parameterisation the potential is in its final form given by
    equation~\eqref{eq:quadratic-qp-potential}, and the associated invariant
    read as: 
    \begin{equation}\label{eq:I2II-linear-combination}
	I_{2,II} = \beta_+\zeta I_{2, IIa} + \alpha_{1+} I_{2,IIb}, 
    \end{equation}
    where $I_{2, IIa}, I_{2,IIb}$ are given
    by~\eqref{eq:I2IIa}-\eqref{eq:I2IIb}.  We can directly check that these are
    \emph{bona fide} invariants, satisfying the condition $I(t+h) = I(t) $.
    Therefore, the first part of the statement is proved. 

    Let us discuss the singular solutions of the system
    \eqref{eq:quadratic-system}.  First, set $\beta_{+}=-\alpha_{2+}M$. In this
    case, the equations
    \eqref{eq:quadratic-system1}-\eqref{eq:quadratic-system2} become: 
    \begin{subequations}\label{eq:V2sing-PDE-system}
            \begin{align}\label{eq:V2sing-PDE-system1}
                    4 \alpha_{2+} A_- \frac{\partial V (A_-, B_-, M )}{\partial B_-} 
                    + 2 \alpha_\pm A_- + \alpha_{1+} + \alpha_{1-} 
                    = 0 
                    \\ \label{eq:V2sing-PDE-system2}
                    2 \alpha_{2+} M \frac{\partial V (A_-, B_-, M )}{\partial B_-}  - \varkappa = 0, 
            \end{align}
                    \begin{equation}\label{eq:V2sing-PDE-system-nonlin}
                            \begin{aligned}
                                                    4\alpha_{2+} (A_-^2 - B_- M )
                    \left(\frac{\partial V}{\partial B_-}\right)^2 
                    + \left( \alpha_\pm A_- M + \varkappa A_- + \alpha_{1-} M \right)
                    \frac{\partial V}{\partial A_- }&
                    \\ + 2 \left( 
                    \alpha_\pm A_-^2 + \alpha_{2+} M^2 
                    + \alpha_{1-} A_- + \varkappa B_- 
                    \right) 
                    \frac{\partial V}{\partial B_-}
                    + (\alpha_{1+} - \alpha_{1-}) A_- 
                    - \varkappa M &=0. 
                            \end{aligned}
                    \end{equation}
    \end{subequations}
    We solve the first PDE \eqref{eq:V2sing-PDE-system1} alone:
    \begin{equation}\label{eq:V-partial-singular1}
            V (A_-, B_-, M) = 
            \frac{\varkappa}{2\alpha_{2+}}
            \frac{B_-}{M}
            + G(M, A_-), 
    \end{equation}
    and use \eqref{eq:V2sing-PDE-system2}-\eqref{eq:V2sing-PDE-system-nonlin}
    as a compatibility conditions. Substituting \eqref{eq:V-partial-singular1}
    into  \eqref{eq:V2sing-PDE-system2}, taking the coefficients and solving
    the resulting linear system with respect to the functions $\alpha_{1-}$ and
    $\varkappa $, we obtain the following set of solutions: 
    \begin{equation}\label{eq:V2-algebr-system-singular}
              \alpha_{1-} = - \alpha_{1+} =0, \quad
                \varkappa = - \alpha_\pm M , 
                \qquad 
               \alpha_{1+},  \alpha_{2+}, \alpha_\pm  =\text{arbitrary}. 
    \end{equation}
    Inserting~\eqref{eq:V2-algebr-system-singular}
    into~\eqref{eq:V-partial-singular1}, and introducing the scaling
    \begin{equation}\label{eq:rescaling2} \alpha \coloneqq
    \frac{\alpha_\pm}{\alpha_{2+}}, \end{equation} we obtain the first singular
    solution~\eqref{eq:V2Is}, together with the corresponding
    invariant~\eqref{eq:I2Is}.  Note that, the denominator $\alpha_{2+} \ne 0$,
    otherwise from $\beta_{+}=-\alpha_{2+}M$ one obtain $\beta_+ =0 $, which is
    the second case.

    In the same fashion, inserting $\beta_+ =0  $ (together with
    \eqref{eq:quadratic-system-alg-alpha2}) into the
    \eqref{eq:quadratic-system1}-\eqref{eq:quadratic-system2}, we arrive at: 
    \begin{subequations}
        \begin{align}\label{eq:V2s-partial-system-beta-0}
                2   \alpha_{2+} M  \frac{\partial  V(A_-, B_-, M)}{\partial A_-} 
         + 4 \alpha_{2+} A_- \frac{\partial  V(A_-, B_-, M)}{\partial B_-} 
         + 2 \alpha_\pm A_- + \alpha_- + \alpha_+ &=0
                 \\
                 \varkappa &= 0, 
        \end{align}
    \end{subequations}
    The solution of \eqref{eq:V2s-partial-system-beta-0} reads as 
    \begin{equation}\label{eq:V2s-partial-solution-beta-0}
        V (A_-, B_-, M) = \frac{\alpha}{2M} A_-^2 
        + \frac{\alpha_{1-} + \alpha_{1+} }{2 \alpha_{2+} M } A_-
        + F \left(M, \, MB_- - A_-^2 \right), 
    \end{equation}
    where we used again the scaling~\eqref{eq:rescaling2}.
    Substituting~\eqref{eq:V2s-partial-solution-beta-0}
    into~\eqref{eq:quadratic-system3}, we obtain a system of algebraic equations
    for the functions $\alpha_{1-}$, $\alpha_{1+}$, $\alpha_{2+} $ and
    $\alpha_{\pm} $. The solutions of this system are s two sets of possible
    values:
    \begin{subequations}
            \begin{align}
             \label{eq:singular-alg-system2b}
                    \alpha_{1-}& =- \alpha_{1+}, \quad  
                    \beta_+ = 0, \quad 
                    \alpha_{\pm} =- 2 \alpha_{2+} , \quad 
                    \varkappa = 0, \qquad 
                    \alpha_{1+},  \, \alpha_{2+}  = \text{arbitrary}; 
            \\
            \label{eq:singular-alg-system2a}
                    \alpha_{1-}& = \alpha_{1+}, \quad 
                    \beta_+ = 0, \quad
                    \varkappa = 0, \qquad  
                    \alpha_{1+}, \, \alpha_\pm, \, \alpha_{2+}  = \text{arbitrary}.
            \end{align}
    \end{subequations}
    The \eqref{eq:singular-alg-system2b} leads again to the already studies
    potential \eqref{eq:V1s}, while the \eqref{eq:singular-alg-system2a}
    provides the second singular potential \eqref{eq:V2IIs} and its invariant
    \eqref{eq:I2IIs}, where 
    \begin{equation}
            \alpha_+ \coloneqq \frac{\alpha_{1+}}{\alpha_{2+}}. 
    \end{equation}

    Finally, we observe that we have to impose the following conditions on the
    functions $G$ and $F$:
    \begin{equation}
            \frac{\partial^{3} G}{\partial Y^{3}}(X,Y) \neq 0, 
            \qquad 
              \frac{\partial^{2} F}{\partial Y^{2}}(X,Y) \neq 0
    \end{equation}
    in order to avoid falling into the already-studied
    cases~\eqref{linear-potential-solution-united},~\eqref{eq:quadratic-qp-potential-2},
    and~\eqref{eq:quadratic-qp-potential}.
\end{proof}

\section{Study of obtained potentials}
\label{sec:sys}

In this section, we study the obtained in propositions~\ref{prop:linear}
and~\ref{prop:quadr-phys} potentials and invariants and identify the obtained
systems, relating them to what is known in the literature.  First, we will
discuss the non-singular potentials, while we will discuss the singular ones
altogether in the final subsection.

\subsection{The potential $V_{1}$}
\label{sect:V1}

Let us consider the potential $V_{1}$ from~\eqref{linear-potential-solution-united}.
In the realization~\eqref{eq:h6reaN}, it gives raise to the following Lagrangian:
\begin{equation}\label{eq:lagrangianV1}
    L_{1} = \sum_{k=1}^{N}\left[ q_{k}(t+h)q_{k}(t)  
    + \alpha_+ \lambda_{k} q_{k}(t) 
+ \frac{\varkappa}{2} q_{k}^{2}(t)\right].
\end{equation}

The associated discrete Euler--Lagrange equations are:
\begin{equation}
    q_{k}(t+h)+\varkappa q_{k}(t) + q_{k}(t-h) = - \alpha_+ \lambda_{k},
    \label{eq:dELV1} 
\end{equation}
while the canonical form of the system is the following:
\begin{subequations}\label{eq:V1-canonical-eoms}
    \begin{align}
        q_{k}(t+h) + p_{k}(t) &= -\varkappa q_{k}(t) - \alpha_+ \lambda_{k},
        \label{eq:canV1q}
        \\
        p_{k}(t+h) &= q_{k}(t).
        \label{eq:canV1p}
    \end{align}
    \label{eq:canV1}%
\end{subequations}

In the realization~\eqref{eq:h6reaN} the invariant $I_{1}$ from~\eqref{linear-invariants-two}
has the following form:
\begin{equation}
    I_{1} =
        \sum_{k=1}^{N}\left[ p_{k}^{2}(t) +q_{k}^{2}(t) 
            +\alpha_{+}\lambda_{k}(p_{k}(t)+q_{k}(t))
        +\varkappa q_{k}(t)p_{k}(t) + \varkappa\frac{\lambda_{k}^{2}}{2} \right].
    \label{eq:I1real}
\end{equation}

The system~\eqref{eq:dELV1} is related to the discrete isotropic harmonic oscillator
\begin{equation}
    q_{k}(t+h) + \varkappa q_{k}(t) + q_{k}(t-h) = 0,
    \quad
    k=1,\dots,N,
    \label{eq:diho}
\end{equation}
through the point transformation:
\begin{equation}
    q_{k}(t) \mapsto q_{k}(t) -\frac{\alpha_{+}\lambda_{k}}{\varkappa+2},
    \quad
    k=1,\dots,N.
    \label{eq:transfdiho}
\end{equation}
This implies that this system is MS, for instance with the $2N-1$ invariants,
constructed with the coalgebra symmetry with respect to the
$\mathfrak{sl}_{2}(\R)$ algebra presented in~\cite[Example 1 and Remark
4.2]{Gubbiotti_EtAl2023.Coalgebra_symmetry_discrete_systems}. Some additional
considerations, generalizations, and a different proof of superintegrability
with a discrete version of the Demkov--Fradkin
tensor~\cite{Demkov1956,Fradkin1965} and its symmetry algebra will be treated
in a subsequent work~\cite{DGL_fradkin}.

We observe that, using the coalgebra symmetry with respect to the $\h_{6}$ we
are able to prove at most that the system is QMS. This proof is a direct proof
involving the direct construction of the invariants quadratic in generators of
the algebra. Since this proof is rather long and technical, we reported it in
\cref{app:V1qms}. We observe that this is consistent with the fact that often
the last invariant of motion is not of coalgebraic origin, see for
instance~\cite{PostRiglioni2015}.

We conclude this subsection with a simple discussion of the continuum limit of
the discrete-time system~\eqref{eq:dELV1}. Note that, as $h\to0^{+}$ we have:
\begin{equation}\label{eq:qh-expansion}
    q_k (t \pm h) = Q_k (t) \pm h \, \dot Q_k(t) + \frac{h^2}{2} \ddot Q_k (t) 
    + \mathcal{O}(h^2).
\end{equation}
Therefore, the equation \eqref{eq:dELV1} takes the form
\begin{equation}\label{eq:ODE2}
    \ddot Q_k(t) + \frac{\varkappa + 2}{h^2}  Q_k(t)   = 
    - \frac{\alpha_+}{h^2} \lambda_k. 
\end{equation}
So, to balance the two sides of the equation~\eqref{eq:ODE2}, we set 
\begin{equation}
    \varkappa =  \omega^2  h^2-2,
    \quad
    \alpha_+ = \gamma h^2. 
    \label{eq:coeff_scalingV1}
\end{equation}
where $\omega $ and $\gamma$ are arbitrary functions of $\vec \lambda$.
Under such a scaling we obtain that as $h\to0^{+}$, the continuum limit 
of the difference equation~\eqref{eq:dELV1} is given by 
\begin{equation}\label{eq:cont-EL-V1}
    \ddot Q_k (t) + \omega^2 Q_k = - \gamma \lambda_k. 
\end{equation}
The same scaling~\eqref{eq:coeff_scalingV1}, applied to the 
discrete Lagrangian~\eqref{eq:lagrangianV1} yields, up to total
derivatives:
\begin{equation}
    L_{1} = -h^{2} \mathcal{L}_{1} + \mathcal{O}(h^{3}),
    \label{eq:dL1contlim}
\end{equation}
where
\begin{equation}\label{eq:cont-Lagrangian-V1}
	\mathcal{L}_1 \coloneqq \sum_{k=1}^N \left[ 
		\frac{1}{2} \dot Q_k^2 (t)  
	   - \frac{1}{2} \omega^2 Q_k^2 (t) 
	   - \gamma \lambda_k Q_k(t) 
	\right] 
\end{equation}
is the Lagrangian generating the equations of
motion~\eqref{eq:cont-EL-V1}.

Using the scaling~\eqref{eq:coeff_scalingV1}, and that the equation of momenta
imply $\dot{Q}_{k}=P_{k}$, we obtain that the continuum limit of the
invariant~\eqref{eq:I1real} is related to the Hamiltonian, associated to the
Lagragian~\eqref{eq:cont-Lagrangian-V1}: 
\begin{equation}\label{eq:cont-Hamiltonian-V1}
	\mathcal{H}_1 \coloneqq   \sum_{k=1}^N \left[
	\frac{1}{2}P_k^2(t) 
	+ \frac{1}{2}\omega^2 Q_k^2(t) 
	+ \gamma \lambda_k Q_k (t)
	\right].  
\end{equation}
through the relation:
\begin{equation}
    I_1 = \vl^{2} + h^2 \left[2 \mathcal{H}_{1} - 
    \frac{\omega^2 \vl^2 }{2}\right] + \mathcal{O}(h^{3}),
    \label{eq:I1exp}
\end{equation}
\emph{i.e.} the Hamiltonian appears at the first non-trivial
order of expansion.

In addition, under the same assumptions, the left and right Casimir functions
$\mathfrak{h}_6$~\eqref{eq:Cash6} in continuum limit have the simple expansion
$C_{[m]}(\vec{q},\vec{p})=h^{2}C_{[m]}(\vec{Q},\vec{P})+\mathcal{O}(h^{3})$,
and
$C^{[m]}(\vec{q},\vec{p})=h^{2}C^{[m]}(\vec{Q},\vec{P})+\mathcal{O}(h^{3})$,
for $m=2,\dots,N$. That is, as expected, the left and right Casimir functions
are inherited unchanged by the continuous system, see for
instance~\cite{ballesterosTwoPhotonAlgebraIntegrable2001a,Ballesteros_EtAl2009.Ndimensional_integrability_twophoton_coalgebra_symmetry}.

\subsection{The potential $V_{2, I}$ } \label{sect:si}  

Let us now consider the potential $V_{2, I} $
\eqref{eq:quadratic-qp-potential-2}, obtained as one of the solutions within
proposition~\ref{prop:quadr-phys}. In realization \eqref{eq:h6reaN}, it can be
rewritten as follows: 
\begin{equation}
	V_{2, I} = \frac{1 }{2}\left[ 
			 \frac{\left(\kappa  +2 \right)}{ \vl^2 }  ( \vl \cdot \vq )^2  
			- \kappa \vq^2 
			\right]. 
\end{equation}

It gives rise to the discrete Lagrangian 
\begin{equation}\label{eq:lagrangianV2I}
	L_{2,I} = \sum_{k=1}^N  \left[ q_k(t+h) q_k(t)  
	+ \frac{\kappa}{2}  q_k^2
		\right] 
		- \frac{(\kappa +2 )}{2 \vl^2   } 
	\left( 
	 \sum_{i=1}^N 
	\lambda_i q_i \right)^2. 
\end{equation}
	
The associated discrete Euler--Lagrange equations are: 	
\begin{equation} \label{eq:dELV2I} 
	q_k(t+h) + \kappa q_k(t) + q_k (t-h)  = 
	\frac{(\kappa+2)}{
	\vl^2 
	} 
	\left(\vl \cdot \vec q(t) \right)\lambda_k  , 
\end{equation}
while the corresponding system in canonical form is 
\begin{subequations}
	\begin{align}
		q_k(t+h) + p_k (t) &= -\kappa q_k(t)  + \frac{(\kappa+2)}{
	\vl^2 
	} 
	\left(\vl \cdot \vec q(t) \right)\lambda_k,  
		\\ 
		p_k (t+h) & = q_k (t). 
	\end{align}
\end{subequations}	

This system possesses two invariants, that can be written in the
realization~\eqref{eq:h6reaN}:
\begin{subequations}
    \begin{align}
        I_{1}^{(s)} &=
        \sum_{k=1}^{N} \lambda_{k}\left( p_{k}-q_{k} \right),
        \label{eq:I1sreal}
        \\
        I_{2,I} &= 	-  \left( 
			1
			+ \frac{\kappa }{2}  \right)
			\frac{\left( (\vl \cdot \vp)^2 + (\vl \cdot \vq)^2 \right)}{\vl^2 }
			+ \kappa \,  \left( (\vq \cdot \vp) - \frac{\vl^2}{2} \right) 
                        + \vp^2 
                        + \vq^2. 
    \end{align}
    \label{eq:invV2I}
\end{subequations}

We now prove that the system~\eqref{eq:dELV2I} is superintegrable:

\begin{theorem}  \label{th:SI}
The system \eqref{eq:dELV2I}, generated by the Lagrangian~\eqref{eq:lagrangianV2I} with the potential $V_{2,I}$, is superintegrable, because it admits the following $2N-3$ functionally independent invariants 
   \begin{equation}\label{SI1-set} 
       	 \mathcal{S}_{2, I} = \left\{ I_{2, I}, I_{1}^{(s)},  C^{[3]}, \dots, C^{[N]}, C_{[3]}, \dots, C_{[N-1]} 
       	 \right\}, 
       \end{equation}  
of which the first $N$ are in involution. 
\end{theorem}
 
\begin{proof}
    From \Cref{thm:h6}, \Cref{prop:quadr-phys} and \Cref{rmk:I1s} we know that
    the system~\eqref{eq:dELV2I} admits the invariants listed
    in~\eqref{SI1-set}. Involutivity of these invariants follows from direct
    computation on $S\h_{6}$ for $I_{2,I}$ and $I_{1}^{(s)}$, and from the
    general theory, see~\Cref{thm:sumup}. So, we only have to prove their
    functional independence. 
 
 First of all, the left and right Casimir functions $\{  C^{[3]}, \dots C^{[N]},C_{[1]},
    \dots, C_{[N-1]}  \}  $  are functionally independent and local except
    $C^{[N]}$, see, for
instance~\cite{Ballesteros_EtAl2009.Super_integrability_coalgebra_symmetry_Formalism_applications}.  Note that in realization \eqref{eq:h6reaN} 
	\begin{equation}
		C^{[N]} = 
		   \vl^2 \vq^2 \vp^2  
        - \vp^2  (\vl \cdot \vq)^2  
        - \vq^2 (\vl \cdot \vp )^2 
        - \vl^2 \left( \vq \cdot \vp 
        \right)^2
        + 2 (\vl \cdot \vq) (\vl \cdot \vp) \left(  \vq \cdot \vp \right) 
		\mathop{=}^{\eqref{eq:Cash6} }
		 C. 
	\end{equation} 
	    Due to the locality, we need to prove that the nonlocal invariants, $I_{2,I}$, 
    $I_{1}^{(s)}$, and $ C$ are functionally independent, that is:
    \begin{align}
\label{eq:RankJacV2I}  
        \rank\left[ \frac{\partial (I_{2,I}, \, I_{1}^{(s)},  C )}{\partial ( \vec q, \vec p)    } \right]= 3.
    \end{align}
    
    Observe that the $I_{2,I}$ and $I_{1}^{(s)}$ have coalgebraic
    origin. So, we can consider them as functions of the generators of algebra $\vec J $
    composed with the $N$-dimensional symplectic realization. 
    That is, we can write the following decomposition of the Jacobi matrix: 
	\begin{equation}
        \label{eq:JacDecV2I}
        \frac{\partial (I_{2,I} , I_{1}^{(s)},  C  )} {\partial (\vec q, \vec p)} = 
        \frac{\partial ( I_{2,I} ,I_{1}^{(s)},  C )} 
        {\partial (A_+, A_-, B_+, B_-, K )} 
        . 
        \frac{\partial(A_+, A_-, B_+, B_-, K )}{\partial (\vec q, \vec p)}, 
    \end{equation}
      where:
    \begin{equation}\label{eq:jacobian-partial-generators-in-realization}
        \frac{\partial(A_+, A_-, B_+, B_-, K ) }{\partial (\vec q, \vec p  ) }
        = 
        \begin{pmatrix}
            \lambda_1 & \lambda_2 & \dots & \lambda_N & 0 & 0 &  \dots & 0 
            \\ 
            0 & 0& \dots & 0 & \lambda_1 & \lambda_2 & \dots &\lambda_N
            \\ 
            2q_1 & 2q_2& \dots & 2q_N & 0 & 0& \dots &0 
                            \\ 
                            0 &0& \dots & 0 & 2p_1 & 2p_2 & \dots & 2 p_N 
                            \\ 
                            p_1 & p_2 & \dots & p_N & 
                            q_1 & q_2 & \dots & q_N 
            \end{pmatrix}_{5 \times 2N},
     \end{equation}
     which clearly has the rank maximal and equal to $5$ for any $N$.
 
 The other factor $ {\partial ( I_{2,I} ,I_{1}^{(s)},   C )} /
        {\partial (A_+, A_-, B_+, B_-, K )}  $ in \eqref{eq:JacDecV2I} has the form     
 	\begin{equation}
 	\label{Jacobian-abstractV2} 
	   \begin{pmatrix}
			- \dfrac{(2  + \kappa ) }{M} A_+  & 
			-\dfrac{(2 + \kappa )  }{M}A_- 
			& 1 & 1 &  \kappa  
			\\
			1 & -1 & 0 & 0 & 0 
			\\ 
			A_- (2K+M) - 2 A_+ B_-
			& A_+ (2K+M) - 2 A_- B_+ 
			& M B_- - A_-^2  
			& M B_+  - A_+^2 
			& 2 A_+ A_- - M (2K + M)
	   \end{pmatrix}. 
 	\end{equation}
 To prove that the matrix \eqref{Jacobian-abstractV2} has a maximal rank $3$, one can observe that, for instance, the following determinant does not vanish: 	
 	\begin{equation}
 		\det \begin{pmatrix}
 			- \dfrac{(2  + \kappa )}{M} A_+  & 1 & 1
 			\\ 
 			1& 0 & 0 
 			\\ 
 			A_- (2K+M) - 2 A_+ B_- &   M B_- - A_-^2  & 
 			M B_+  - A_+^2  
 		\end{pmatrix} 
 		= A_+^2 - A_-^2 + M( B_+ + B_-). 
 	\end{equation}
 	
 	    Since in general for matrix products it holds
    that $\rank (AB)= \min ( \rank(A), \rank(B))$, if $B$ has maximal rank,
    we obtain \eqref{eq:RankJacV2I}, thereby proving the statement. 
\end{proof} 

Let us now discuss the continuum limit of the system~\eqref{eq:dELV2I},
Applying a scaling, similar to the one defined by formula~\eqref{eq:coeff_scalingV1}:
\begin{equation}
    \kappa =  \omega^2  h^2-2,
    \label{eq:coeff_scalingV2I}
\end{equation}
and expanding the functions $q_{k}(t)$ as in~\eqref{eq:qh-expansion} for 
$h \to 0^+ $ to the system \eqref{eq:dELV2I}, one arrives at the following 
differential equation:
\begin{equation}\label{eq:cont-EL-eqs-V2I}
	\ddot Q_k(t) + \omega^2 Q_k(t) 
	=  \omega^2 \frac{\lambda_k}{\vl^2} 
	\sum_{i=1}^N \lambda_i Q_i.
\end{equation}
Analogously, the continuum limit of the discrete
Lagrangian~\eqref{eq:lagrangianV2I} is up to total derivative: 
\begin{equation}
    L_{2,I}  = -h^{2} \mathcal{L}_{2,I} +\mathcal{O}( h^{3} ),
    \label{eq:L2Icl}
\end{equation}
where
\begin{equation}
    \label{eq:L2Icont}
	\mathcal{L}_{2,I} \coloneqq \frac{1}{2} \sum_{k=1}^N \left[
		 \dot Q_k^2(t) 
            \right]
		-   \omega^2 Q_k^2(t) 
                +\frac{ \omega^2 }{2 \vl^2} (\vl \cdot \vq)^2,
\end{equation}
is exactly the Lagrangian of the system~\eqref{eq:cont-EL-eqs-V2I}.

The invariants $I_{2,I} $ and $I_{1}^{(s)}$ in continuum limit take the form
\begin{subequations}
    \begin{align}
        I_{2,I} & \vl^2 
            + h^2 \left[ 2 \mathcal{H}_{2,I} - \omega^2 \frac{\vl^2}{2 }\right]
            + \mathcal{O}(h^3), 
        \\ 
        I_{1}^{(s)} &=  - h \left(\vl \cdot \vP  \right) 
            + \mathcal{O}(h^2). 
    \end{align}
\end{subequations}
where 
\begin{equation}
    \label{eq:H2Icont}
	\mathcal{H}_{2,I} \coloneqq \sum_{k=1}^N \frac{1}{2}
	\left[ P_k^2 (t) +  \omega^2 Q_k^2(t) \right],
        - \frac{\omega^2}{\vl^2} (\vl \cdot \vec{Q})^2,
\end{equation}
is the Hamiltonian associated to the Lagrangian~\eqref{eq:L2Icont}.

\begin{remark}
    We observe that the canonical equations of~\eqref{eq:H2Icont}
    are given by:
    \begin{subequations}
            \begin{align}
                    \dot P_k (t) & = \omega^2 \left[ Q_k(t) - \frac{(\vl \cdot \vec{Q}(t))}{\vl^2} \lambda_k  \right], 
                    \\ 
                    \dot Q_k(t) & = P_k(t), 
            \end{align}
    \end{subequations} 
    and their explicit solution is particularly simple:
    \begin{subequations}
        \begin{align}
            Q_k(t) &= Q_{k,0} \sin \left( \omega  \sqrt{ 1 - \frac{\lambda_k}{\vl^{2}} }t +\varphi_{k}\right),
            \label{eq:QkV2Isol}
            \\
            P_k(t) &= Q_{k,0}\omega  \sqrt{ 1 - \frac{\lambda_k}{\vl^{2}} } 
            \cos \left( \omega  \sqrt{ 1 - \frac{\lambda_k}{\vl^{2}} }t +\varphi_{k}\right) 
            \label{eq:PkV2Isol}
        \end{align}
        \label{eq:V2Isol}
    \end{subequations}
    where $Q_{k,0}$ and $\varphi_k$ are integration constants. 
    \label{rmk:solution}
\end{remark}

In this case, $(\vl\cdot \vec{P})$, $\mathcal{H}_{2,I} $, and $C$ are
functionally independent. Moreover, we have 
\begin{equation}
    \{ I_{1}^{(s)}, \mathcal{H}_{2,I} \}=0. 
\end{equation}
Therefore, we can choose the set 
\begin{equation}
    \left\{  (\vl\cdot\vP), \,  \mathcal{H}_{2,I}, \, C^{[3]},\dots,C^{[N]},
        C_{[3]},\dots,C_{[N-1]}
	\right\} 
\end{equation}
as a set of functionally independent invariants for the continuous system,
of which the first $N$ are in involution.

\begin{remark}
    We remark that the continuous Hamiltonian system given by
    $\mathcal{H}_{2,I}$ is in the family of systems with invariance with
    respect to the generator $A_{+}$ studied in
    \cite{Ballesteros_EtAl2009.Ndimensional_integrability_twophoton_coalgebra_symmetry},
    see also~\cite[\S 4.5.2]{blascosanzIntegrabilidadSistemasNo2009}. Indeed,
    $I_{1}^{(s)} = -hA_{+}+\mathcal{O}(h^{2})$ at first order is a generator of
    the coalgebra. Note that in the discrete setting, no such example of
    invariance with respect to a single generator arose.
    \label{rmk:blascoV2I}
\end{remark}

\subsection{The potential $V_{2,II} $}

Let us now consider another solution, obtained in \Cref{prop:quadr-phys},
namely the  potential $V_{2,II}$. Given the
realization~\eqref{eq:quadratic-qp-potential}, it can be rewritten as 
\begin{equation}\label{eq:V2II-rescaled}
	V_{2,II} = - \frac{\eta}{2} (\vl \cdot \vq)^2 
	- \zeta (\vl \cdot \vq) 
	- \frac{\kappa}{2} \vq^2. 
\end{equation}
Therefore, it produces the following Lagrangian: 
\begin{equation}
	\label{eq:lagrangianV2II}
    L_{2,II} = \sum_{k=1}^{N}\left[ q_{k}(t+h)q_{k}(t)  
    + \zeta  \lambda_k q_k + 
    \frac{\kappa}{2} q_k^2
    \right] 
    + \frac{\eta }{2} (\vl \cdot \vq)^2. 
\end{equation} 
The associated system of discrete Euler--Lagrange equations is: 
\begin{equation}\label{eq:dELV2II} 
	q_k(t+h) + q_k(t-h) = 
	- \left[ \eta  (\vl \cdot \vq) + \zeta 
	\right] \lambda_k 
	- \kappa q_k, 
\end{equation}
which can also be rewritten in canonical form
\begin{subequations}\label{eq:can-eqs-for-V2II} 
	\begin{align}
		q_k(t+h) + p_k(t) &= 
		-  \left[ \eta  (\vl \cdot \vq) + \zeta 
	\right] \lambda_k 
	- \kappa q_k,
	\\ 
	p_k(t+h) & = q_k(t).  
\end{align}  
\end{subequations}

\begin{remark}
    We observe that the system generated by the discrete
    Lagrangian~\eqref{eq:lagrangianV2II} is a slight generalization of the
    system given in~\cite[Example
    2]{Gubbiotti_EtAl2023.Coalgebra_symmetry_discrete_systems}, which is
    obtained for $\zeta=-1$, $\kappa=-\alpha_{0}$, and $\eta=-\alpha_{1}$.
    \label{rmk:oldpotential}
\end{remark}

From~\Cref{prop:quadr-phys} we have two invariants, which written
in the realization~\eqref{eq:h6reaN} can be written as:
\begin{subequations}
    \begin{align}
        I_{2,IIa} &= -  \kappa (\vl \cdot \vq)  (\vl \cdot \vp)  
		 +   \kappa \left((\vq \cdot \vp) - \frac{\vl^2}{2} \right)  \vl^2
		 -  (
		  (\vl \cdot \vq)^{2} 
		 +(\vl \cdot \vp)^{2} )  
		 + ( \vq^2    
		 + \vp^2    ) \vl^2,  
        \label{eq:I2IIarea}
        \\
        I_{2,IIb} &= (\eta \vl^2 + \kappa) (\vl \cdot \vq) (\vl \cdot \vp) 
		 + \zeta \vl^2 \left((\vl \cdot \vp)  + (\vl \cdot \vq) \right)
		 +(\vl \cdot \vp)^2 + (\vl \cdot \vq)^2. 
        \label{eq:I2IIbrea}
    \end{align}
    \label{eq:IV2IIrea}
\end{subequations}

So, let us prove that this system~\eqref{eq:dELV2II} is superintegrable: 

\begin{theorem} \label{th:SI1}
	The system~\eqref{eq:dELV2II}, generated by the Lagrangian~\eqref{eq:lagrangianV2II}, is superintegrable because it admits a set of $2N-3$ functionally independent invariants
		   \begin{equation}\label{QI-set} 
       	 \mathcal{S}_3 = \left\{ I_{2,IIa}, I_{2,IIb}, C^{[3]}, \dots, C^{[N]}, C_{[3]}, \dots, C_{[N-1]} 
       	 \right\}, 
       \end{equation}  
	 of which the first $N$ are in involution.   
  \end{theorem}

\begin{proof}
    Given \Cref{thm:h6} and \Cref{prop:quadr-phys}, we have that the
    system~\eqref{eq:dELV2II} admits the set of invariants \eqref{QI-set}.
    Involutivity of these invariants follows from direct computation on
    $S\h_{6}$ for $I_{2,IIa}$ and $I_{2,IIb}$, and from the general theory,
    see~\Cref{thm:sumup}. So, we only have to prove their functional
    independence. 

As in the case of \Cref{th:SI}, it is sufficient to prove that 
\begin{equation}\label{eq:rank-Jac-for-th-SI1}
		\rank \left[	\frac{\partial ( I_{2,IIa}, I_{2,IIb} , C )}{\partial ( \vec q, \vec p)    }  
		 \right] =3. 
\end{equation}  
We use the similar \eqref{eq:JacDecV2I} to decomposition: 
\begin{equation}
        \label{eq:JacDecV2II}
        \frac{\partial  ( I_{2,IIa}, I_{2,IIb} , C ) } {\partial (\vec q, \vec p)} = 
        \frac{\partial ( I_{2,IIa}, I_{2,IIb} , C ) } 
        {\partial (A_+, A_-, B_+, B_-, K )} 
        . 
        \frac{\partial(A_+, A_-, B_+, B_-, K )}{\partial (\vec q, \vec p)}, 
    \end{equation}
Performing straightforward computations, we arrive at the following expression for $    \dfrac{\partial ( I_{2,IIa}, I_{2,IIb} , C ) } 
        {\partial (A_+, A_-, B_+, B_-, K )}  $:
\begin{equation}
	\begin{pmatrix}
		- \kappa A_- - 2 A_+  
		 & - \kappa A_+ - 2 A_- & M & M & \kappa M, 
		 \\ 
		 (\eta M +\kappa) A_- +2A_+ + \zeta M 
		 & 
		 2 A_- + (\eta M + \kappa ) A_+ + \zeta M 
		 & 0 & 0 & 0 
		 \\ 
		 2 (A_- K - A_+ B_- ) + A_- M
		 & 2 (A_+ K + A_- B_+ ) + A_+ M 
		 & B_- M - A_-^2 
		 & B_+ M - A_+^2 
		 & 2 (A_+ A_- - KM) - M^2
	\end{pmatrix}
\end{equation}
Since, for example, 
\begin{multline}
	\det 	\begin{pmatrix}
		- \kappa A_- - 2 A_+  
		  & M  & M  \\ 
		 (\eta M +\kappa) A_- +2A_+ + \zeta M 
		 & 0 & 0 
		 \\ 
		 2 (A_- K - A_+ B_- ) + A_- M
		 & B_- M - A_-^2 
		 & B_+ M - A_+^2 
	\end{pmatrix} 
	= - M \left( \eta M A_- + \kappa A_- + \zeta M + 2 A_+ \right) \times  
	\\ \times 
	\left(A_-^2 - A_+^2 - M B_- + MB_+  \right) 
	\ne 0,
\end{multline}
we can conclude that its rank is $3$. Then, applying again the formula for the rank of the matrix product to \eqref{eq:JacDecV2II}, we obtain \eqref{eq:rank-Jac-for-th-SI1} and prove the theorem. 
\end{proof}

Let us now discuss the continuum limit of the system, generated by the Lagrangian \eqref{eq:lagrangianV2II}. We use the following scaling of the coefficients: \begin{equation}
	 \kappa = \omega^2 h^2 - 2, 
	 \qquad
	 \zeta = \gamma h^2,
	 \qquad 
	 \eta = \delta h^2. 
\end{equation}
Therefore, the continuum version of the system \eqref{eq:dELV2II} is given by: 
\begin{equation}\label{eq:cont-EL-eqs-V2II}
	\ddot Q_k (t) + \omega^2 Q_k(t) 
	+ \lambda_k  \left[ 
	 \gamma  
	+ \delta 
	\left( \vl \cdot \vec{Q} (t)  \right)
	\right] 
	= 0, 
	\qquad 
	k=1,\dots,N. 
\end{equation}
The Lagrangian, of the system \eqref{eq:cont-EL-eqs-V2II},
\begin{equation}
	\mathcal{L}_{2, II} \coloneqq 
	\sum_{k=1}^N \left[ 
	 \frac{1}{2} \dot Q_k^2 (t)
	 - \frac{\omega^2}{2} Q_k^2 (t)
	 - \gamma \lambda_k Q_k (t) 
	\right] 
	- \frac{\delta}{2} \left( \vl \cdot \vQ (t) \right)^2,   
        \label{eq:L2IIcont}
\end{equation}
can be retrieved by the $h^{2}$ coefficient of the Taylor
expansion of the discrete Lagrangian $L_{2,II}$~\eqref{eq:lagrangianV2II}.

In turn, the continuum limits of the invariants $I_{2,IIa}$ and $I_{2,IIb}$
are 
\begin{subequations}
    \label{eq:I2expansion}
	\begin{align}
		I_{2,IIa} & =  \vl^4 + 
                h^2 I_{2,IIa}^{(2)} + \mathcal{O} (h^3), 
		\\ 
                I_{2,IIb} &= h^2 I_{2,IIb}^{(2)} + \mathcal{O} (h^3),
	\end{align}
\end{subequations}
where:
\begin{subequations}
	\begin{align}
            I_{2,IIa}^{(2)} &\coloneqq   
		\vl^2 \left( \vP^2(t)  + \omega^2 \vQ^2(t)  -
		\frac{\omega^2 \vl^2}{2}
		\right)
		- \omega^2 \left(\vl \cdot \vQ(t)\right)^2
		- \left( \vl \cdot \vP(t) \right)^2, 
		\\ 
                I_{2,IIb}^{(2)} &\coloneqq  
			\left(\vl \cdot \vP(t)\right)^2 
			+ \omega^2 \left(  \vl \cdot \vQ (t) \right)^2  
			+ \vl^2 \left[
			\delta \left( \vl \cdot \vQ (t) \right)^2 
			+ 2 \gamma \left( \vl \cdot \vQ (t) \right)
			\right] .
	\end{align}
\end{subequations}

Note that the Hamiltonian corresponding to the Lagrangian~\eqref{eq:L2IIcont}
\begin{equation}
	\mathcal{H}_{2,II} \coloneqq \sum_{k=1}^N \left[ 
		\frac{P_k^2(t)}{2} + \frac{\omega^2 Q_k^2(t)}{2} 
		+ \gamma \lambda_k Q_k 
	\right] 
	+ \frac{\delta}{2} \left( \vl \cdot \vQ (t) \right)^2. 
\end{equation}
is not given in the expansion~\eqref{eq:I2expansion}.
However, those invariants are not functionally independent from the Hamiltonian
because 
\begin{equation}
    {I}_{2,IIa}^{(2)} +  {I}_{2,IIb}^{(2)} = 
	\vl^2 \left[ \vP^2 (t) + \omega^2 \vQ^2 
	+ 2 \gamma \left(\vl \cdot \vQ  \right)
	+ \delta \left( \vl \cdot \vQ  \right)^2 
	- \frac{\omega^2 \vl^2 }{2}
	\right]  
        = \vl^2 \left[2 \mathcal{H}_{2,II} - \frac{\omega^2 \vl^2 }{2} \right], 
\end{equation}
that is $ \mathcal{H}_{2,II}$ is a linear combination of ${I}_{2,IIa}^{(2)}$
and ${I}_{2,IIb}^{(2)}$. Moreover, we have
\begin{equation}
    \{ {I}_{2,IIa}^{(2)} \, ,  \, \mathcal{I}_{2,IIb}^{(2)}\}=
    \{{I}_{2,IIa}^{(2)} \, , \, \mathcal{H}_{2,II} \}=
    \{ {I}_{2,IIb}^{(2)}  \, , \, \mathcal{H}_{2,II}  \}=0, 
\end{equation}
Therefore, noting that the ${I}_{2,IIa}$ and  $\mathcal{H}_{2,II}$ are also
functionally independent of $C$, we can choose the set 
\begin{equation}
    \left\{  {I}_{2,IIa}, \,  \mathcal{H}_{2,II}, \, C^{[3]},\dots,C^{[N]},
        C_{[3]},\dots,C_{[N-1]}
	\right\} 
\end{equation}
as a set of functionally independent invariants for the continuous system,
of which the first $N$ are commuting.

\subsection{Singular potentials $ V_{1}^{(s)} $, $V_{2,I}^{(s)}$ and $V_{2,II}^{(s)} $}

Let us consider the obtained singular potentials $V_{1}^{(s)}$~\eqref{eq:V1s},
$V_{2,I}^{(s)}$~\eqref{eq:V2Is}, and $V_{2,II}^{(s)}$~\eqref{eq:V2IIs}. We
recall that the system~\eqref{eq:sysfin} with those potentials become
quasi-integrable.

Let us write the Lagrangians associated to the three potentials:
\begin{subequations}
    \begin{align}
	L_{1}^{(s)} &\coloneqq 
        \sum_{k=1}^{N} \left[q_{k}(t+h)q_{k}(t)  
        - q_{k}^{2} (t) \right]
        - F \left( \vl^2, \ \sum_{1\leq i<j\leq N}\left( \lambda_{i}q_{j}(t)-\lambda_{j}q_{i}(t) \right)^{2} \right),  
    \label{eq:V1S-lagrangian}
        \\
        \label{eq:V2IS-lagrangian}
	L_{2,I}^{(s)} &\coloneqq 
        \sum_{k=1}^{N} \left[q_{k}(t+h)q_{k}(t)  
        + \frac{\alpha}{2} q_{k}^{2}(t) \right]
	- G \left(\vl^2, \ \sum_{j=1}^N \lambda_j q_j (t) \right),
        \\
	\label{eq:V2IIS-lagrangian}
	L_{2,II}^{(s)} &\coloneqq \sum_{k=1}^{N} \left[q_{k}(t+h)q_{k}(t)  
        + \frac{\alpha}{2} q_{k}^{2}(t) +\frac{\alpha_{+}}{\vl^{2}}\lambda_{k}q_{k}(t) \right]
	- F \left( \vl^2, \ \sum_{1\leq i<j\leq N}\left( \lambda_{i}q_{j}(t)-\lambda_{j}q_{i}(t) \right)^{2} \right).
    \end{align}
\end{subequations}
We omit the explicit expression of the associated Euler--Lagrange equations,
and equation in canonical form for the sake of brevity.

\begin{remark}
    The Lagrangian $L_{2,I}^{(s)}$ in formula~\eqref{eq:V2IS-lagrangian} is a
    generalization of the system presented in~\cite[Example
    3]{Gubbiotti_EtAl2023.Coalgebra_symmetry_discrete_systems} which is
    retrieved when the function $G(X,Y)$ is a third degree polynomial in $Y$.
    \label{rmk:partcase}
\end{remark}

In the realization~\eqref{eq:h6reaN}, the canonical systems associated to these
Lagrangian $L_{1}^{(s)}$~\eqref{eq:V1S-lagrangian} admits the invariant
$I_{1}^{(s)}$ presented in formula~\eqref{eq:I1sreal}, while the canonical
systems associated to the Lagrangians
$L_{2,I}^{(s)}$~\eqref{eq:V2IS-lagrangian} and
$L_{2,II}^{(s)}$~\eqref{eq:V2IIS-lagrangian} have the following invariants:
\begin{subequations}  \label{eq:singular-invariants-repeated}
\begin{align}
	  	 I_{2,I}^{(s)}&= 
	  (\vl \cdot \vp) (\vl \cdot \vq) 
	+\frac{1}{\alpha} 
	\left(
	(\vl \cdot \vq )^{2} 
	+ (\vl \cdot \vp)^{2}
	-   \vq^2 \vl^2 
	- \vp^2 \vl^2 \right) 
	-   \left( (\vq \cdot \vp) - \frac{\vl^2}{2} \right) \vl^2, 
	\\ 
		I^{(s)}_{2,II} & = (\vl \cdot \vp)^2  
                	 + (\vl \cdot \vq)^2 
                         +\alpha (\vl \cdot \vp)  (\vl \cdot \vq) 
                         + \alpha_{+} \left( (\vl \cdot \vp) +  (\vl \cdot \vq) \right).  
\end{align}
\end{subequations}
respectively.

Then, the following statement holds:

\begin{theorem}\label{th:V1s}
    The canonical systems associated to the
    Lagrangians~\eqref{eq:V1S-lagrangian}, \eqref{eq:V2IS-lagrangian}, and
    \eqref{eq:V2IIS-lagrangian} are quasi-integrable because they possess the
    following sets of $2N-4$ invariants respectively
    \begin{subequations}
            \begin{align}
                    \label{eq:V1S-set-of-invariants}
                    \mathcal{S}_1^{(s)} & \coloneqq 
                    \left\{ 
                     I_{1}^{(s)}, C^{[3]}, \dots, C^{[N]}, C_{[3]}, \dots, C_{[N-1]} 
                    \right\},
                    \\ 
                    \mathcal{S}_{2,I}^{(s)}
                    &\coloneqq 
                    \left\{ 
                     I_{2,I}^{(s)}, C^{[3]}, \dots, C^{[N]}, C_{[3]}, \dots, C_{[N-1]} 
                    \right\}, 
                            \\ 
                    \mathcal{S}_{2,II}^{(s)}
                    &\coloneqq 
                    \left\{ 
                     I_{2,II}^{(s)}, C^{[3]}, \dots, C^{[N]}, C_{[3]}, \dots, C_{[N-1]} 
                    \right\}
            \end{align}
    \end{subequations}
    of which the first $N-1$ are in involution.  
\end{theorem}

\begin{proof}
    We will prove the statement for the first system, namely the
    one associated to the Lagrangian \eqref{eq:V1S-lagrangian},
    since the proof for the other two cases will be analogous.

    From the  \Cref{thm:h6} and the \Cref{prop:linear}, we know that the system
    associated with the Lagrangian \eqref{eq:V1S-lagrangian}, admits the set of
    invariants~\eqref{eq:V1S-set-of-invariants}. As discussed previously,
    involutivity follows from \Cref{thm:sumup}. So, as before we are left to
    prove the functional independence of the invariants.

    We have the following explicit form of the Jacobian 
	$	{\partial (I_1^{(s)}, C) }/ {\partial  (A_+, A_-, B_+, B_-, K )} $ : 
	\begin{equation}
		\begin{pmatrix}
			1 & - 1 & 0 & 0 & 0 
			\\ 
			A_-M + 2 (A_- K - A_+ B_-) 
			& 
			A_+ M + 2 (A_+ K - A_- B_+ ) 
			& 
			B_- M - A_-2 
			& 
			B_+ M - A_+^2 
			& 
			2( A_+ A_- - KM) - M^2 
		\end{pmatrix}
	\end{equation}
	Since	
	\begin{equation}
	\det 
		\begin{pmatrix}
			1  & 0 
			\\ 
			A_-M + 2 (A_- K - A_+ B_-) 
			& B_- M - A_-^2 
		\end{pmatrix}
		= B_- M - A_-^2 \ne 0,
	\end{equation}
	we can claim that 
\begin{equation}
	\rank \left[ 
	\frac{\partial (I_1^{(s)}, C) } {\partial  (A_+, A_-, B_+, B_-, K )} 
	\right]  = 2, 
\end{equation}
    and, reasoning as in the proofs of the \Cref{th:SI} and \Cref{th:SI1}, we
    obtain the statement. 
\end{proof}

Let us now conclude this section by discussing the relationship between these three QI
systems with some classes of continuous Hamiltonians considered by Blasco in his
Ph.D.\ thesis~\cite{blascosanzIntegrabilidadSistemasNo2009}.  Let us start from
$L_{1}^{(s)}$~\eqref{eq:V1S-lagrangian}.  Under the assumption that the
function $F(X,Y)$ scales as $h^{2}$, \emph{i.e.} $F(X,Y)=h^{2} f(X,Y)$ we
obtain, up to total derivatives, the following continuous Lagrangian
\begin{equation}
    \mathcal{L}_{1}^{(s)} \coloneqq 
    \frac{1}{2}\sum_{k=1}^{N} \dot{Q}_{k}^{2}(t)
    +f \left( \vl^2, \ \sum_{1\leq i<j\leq N}\left( \lambda_{i}Q_{j}(t)-\lambda_{j}Q_{i}(t) \right)^{2} \right),  
        \label{eq:L1Iscont}
\end{equation}
as the (dominant) $h^{2}$ coefficient of the Taylor expansion of the discrete
Lagrangian $L_{1}^{(s)}$~\eqref{eq:V1S-lagrangian}. With an analogous
reasoning, and the scalings $\alpha=h^{2}\omega-2$, and $G(X,Y)=h^{2}g(X,Y)$
we obtain up to total derivatives, the following continuous Lagrangian
\begin{equation}
    \mathcal{L}_{2,I}^{(s)} \coloneqq 
    \frac{1}{2}\sum_{k=1}^{N} \left[\dot{Q}_{k}^{2}(t) -\omega^{2}Q_{k}^{2}(t)\right]
    +g \left( \vl^2, \ \sum_{j=1}^{N}\lambda_{j}Q_{j}(t) \right),  
        \label{eq:L2Iscont}
\end{equation}
as the (dominant) $h^{2}$ coefficient of the Taylor expansion of the discrete
Lagrangian $L_{2,I}^{(s)}$~\eqref{eq:V2IS-lagrangian}.  Finally, with the
scalings $\alpha=h^{2}\omega-2$, $\alpha_{+}=h^{2}\gamma$, and
$F(X,Y)=h^{2}f(X,Y)$ we obtain up to total derivatives, the following
continuous Lagrangian
\begin{equation}
    \mathcal{L}_{2,II}^{(s)} \coloneqq 
    \frac{1}{2}\sum_{k=1}^{N} \left[\dot{Q}_{k}^{2}(t) -\omega^{2}Q_{k}^{2}(t)
    -\frac{\gamma}{\vl^{2}} \lambda_{k}Q_{k}(t)\right]
    +f \left( \vl^2, \ \sum_{1\leq i<j\leq N}\left( \lambda_{i}Q_{j}(t)-\lambda_{j}Q_{i}(t) \right)^{2} \right),  
        \label{eq:L2IIscont}
\end{equation}
as the (dominant) $h^{2}$ coefficient of the Taylor expansion of the discrete
Lagrangian $L_{2,I}^{(s)}$~\eqref{eq:V2IS-lagrangian}.

Now, it is evident that all these systems are particular cases of the Lagrangian
form of the Hamiltonian system 
\begin{equation}
     \mathcal{H}^{(N)} = \frac{1}{2} \sum_{i=1}^N P_i^2 
     + \delta_1 \sum_{i=1}^N Q_i^2 
     + \mathcal{F} \left(\vl^{2} \sum_{1 \leq i <j }^N 
     \left(\lambda_j Q_i - \lambda_i Q_j \right)^2
     \right)
     + \mathcal{G} \left(\vl^{2}, \sum_{i=1}^N \lambda_i Q_i  \right), 
    \label{eq:blasconew}
\end{equation}
for different values of the parameter $\delta_{1}$, and of the two arbitrary
functions $\mathcal{F}$ and $\mathcal{G}$, see \cref{tab:blasco}.
The systems~\eqref{eq:blasconew} were
considered by Blasco in~\cite[Chap. 5]{blascosanzIntegrabilidadSistemasNo2009},
where it was proved that they are Liouville integrable, and superintegrable. 
So, the continuous systems associated to the Lagrangians~\eqref{eq:L1Iscont},~\eqref{eq:L2Iscont}, and~\eqref{eq:L2IIscont} are Liouville integrable, since they admit one
additional invariant, \emph{i.e.} the Hamiltonian itself.

\begin{table}
    \centering
    \begin{equation*}
        \begin{array}{cccc}
            \toprule
            \text{Lagrangian} & \delta_{1} & \mathcal{F} & \mathcal{G}
            \\
            \midrule
            \mathcal{L}_{1}^{(s)} & 0 & \text{arbitrary} & 0
            \\ 
            \mathcal{L}_{2,I}^{(s)} & \omega^{2}/2 & 0 & \text{arbitrary}
            \\
            \mathcal{L}_{2,II}^{(s)} & \omega^{2}/2 & \text{arbitrary} & \gamma Y/X
            \\
            \bottomrule
        \end{array}
    \end{equation*}
    \caption{Identification of the continuum limits of the singular Lagrangians
    with the Hamiltonian systems~\eqref{eq:blasconew}.}
    \label{tab:blasco}
\end{table}

\section{Conclusions and outlook}
\label{sec:conclusions}

\begin{table}[h!] 
\renewcommand{\arraystretch}{3}  
    \centering
    \begin{tabular}{ccp{2.5cm}p{3.7cm}}
    	 \toprule
         \textbf{Invariant}  & \textbf{Potential} $V(A_-, B_-, M) $ & \textbf{Type} & \textbf{Notes}
        \\
        \midrule
       \makecell{linear \\ non-singular }  & $V_{1} = - \alpha_+ A_- - \dfrac{\varkappa}{2} B_- $ 
       & MS  & IHO up to linear transformations, see  \cref{sect:V1},~\cite{Gubbiotti_EtAl2023.Coalgebra_symmetry_discrete_systems,GubLat_sl2} 
         \\
         \midrule
         \makecell{linear \\ singular }  & $ V_1^{(s)}= B_{-} + F \left(M, \ M B_{-}- A_-^2 \right) $  & QI &  \Cref{th:V1s} 
         \\  
         \midrule
          \makecell{quadratic \\ in $(\vq, \vp) $ \\ non-singular  }  & $ V_{2,I}  =
		\dfrac{1 }{2}\left[ 
			 \dfrac{\left(\kappa  +2  \right)}{ M} A_-^2 
			- \kappa B_-
			\right] $ & Superintegrable  & Theorem~\ref{th:SI} 
		\\ 
                \midrule
                   \makecell{quadratic \\ in $(\vq, \vp) $ \\ non-singular  }  & $ 	V_{2,II}   =-  \dfrac{\eta}{2} A_-^2 - \zeta A_- 
- \dfrac{\kappa}{2}  B_-  $ 
& Superintegrable  & \Cref{th:SI1},~\cite{Gubbiotti_EtAl2023.Coalgebra_symmetry_discrete_systems}
		\\ 
                \midrule
                \makecell{quadratic \\ in $(\vq, \vp) $ \\ singular  }   & $ V_{2,I}^{(s)} =  - \dfrac{\alpha}{2 } B_- + G (M, A_-)  $
		& QI &   \Cref{th:V1s}
        \\
        \midrule
          \makecell{quadratic \\ in $(\vq, \vp) $ \\ singular  } & $ 	V_{2,II}^{(s)} = 
          - \dfrac{\alpha }{2} B_-
        	- \dfrac{\alpha_{+}}{  M} A_- 
                + F \left(M, M B_{-} - A_-^2 \right) $ & QI &   \Cref{th:V1s} 
                \\
                \bottomrule
    \end{tabular}
    \caption{Summary of all systems obtained}
    \label{tab:resuming}
\end{table}     

In this paper, we gave necessary and sufficient conditions for a system in
quasi-standard form to admit coalgebra symmetry with respect to the
Lie--Poisson algebra $\h_{6}$. This extends a previous work~\cite{GubLat_sl2},
where the Lie--Poisson algebra $\mathfrak{sl}_{2}(\R)$, which is a subalgebra
of $\h_{6}$, was considered. In particular, we discovered that differently from
the case of $\mathfrak{sl}_{2}(\R)$ systems in quasi-standard form admitting
such property reduces to a system in standard form, see~\Cref{thm:h6}. From the
technical point of view, the proof of~\Cref{thm:h6} is much more complicated
and long with respect to its $\mathfrak{sl}_{2}(\R)$ analog, and required the
introduction of different concepts. This confirms once again the fact that the
discrete-time setting is inherently more complicated and ``rigid'' than the
continuous-time one. So, this finding is an important step forward on the way
in adapting the coalgebra symmetry approach and its applications to discrete
integrable systems, as initiated
in~\cite{Gubbiotti_EtAl2023.Coalgebra_symmetry_discrete_systems}.

Then, we discussed the existence of coalgebraic invariants for the systems
admitting coalgebra symmetry with respect to $\h_{6}$ with a direct search of
invariants for the system on the generators of the coalgebra~\eqref{eq:sysfin}.
This approach yielded a good zoology of systems in an arbitrary number of
degrees of freedom, including a MS one, two superintegrable ones, and three QI
ones. We identified our findings with systems known in the literature, both on
the discrete and the continuous level. We remark that the three discrete-time
QI systems are, up to our knowledge, new.

This direct search provided some interesting insights into the approach to the
problem, but it is unfortunately limited by the technical difficulties arising
with the growth of the degree of the invariant. However, the existence of
classes of QI systems depending on arbitrary functions pave the way to refined
searches for integrable subcases. We underline that such search program can be
now pursued using different methods, \emph{e.g.} an
adaptation of Nevanlinna theory to systems of first-order difference
equations~\cite{Ablowitz_et_al2000}, or algebraic methods based on construction
of commuting elements in the symmetric algebra of $\h_{6}$, see for
instance~\cite{CampoamorStursberg_etal2023,CampoamorStursberg_etal2024arXiv,CampoamorStursberg_etal2023JPhysConfSer}.
Other possibilities can involve the construction of systems associated to the
Lie subalgebras of $\h_{6}$, which will naturally be QI, and then look for the
missing invariant with a different method. That is, an adaptation of the
\emph{subalgebra integrability method} proposed
in~\cite{blascosanzIntegrabilidadSistemasNo2009}. 

We conclude, by remarking that there are many other open questions on the
application of the $\h_{6}$ coalgebra symmetry to discrete integrable systems.
For instance, in~\cite{Gubbiotti_EtAl2023.Coalgebra_symmetry_discrete_systems}
it was shown that the system:
\begin{equation}
    q_{k}(t+h)+q_{k}(t)+q_{k}(t-h) = \frac{\alpha}{\vec{q}^{2}(t)}q_{k}(t)+\lambda_{k},
    \label{eq:dPIN}
\end{equation}
is (super)integrable through the coalgebra symmetry method. However, an
algebraic explanation of the origin of the two additional invariants of this
system has not been given yet.

\section*{Acknowledgements}

PD and GG acknowledge the support of the research project Mathematical Methods
in NonLinear Physics (MMNLP), Gruppo 4-Fisica Teorica of INFN.  PD also
acknowledges the support from the Ph.D. program of the Universit\`a degli Studi
di Udine. GG's research was also partially supported by GNFM of the Istituto
Nazionale di Alta Matematica (INdAM).

\appendix

\section{Proof that the system~\eqref{eq:dELV1} is QMS through coalgebraic invariants}
\label{app:V1qms}

In this appendix, we prove that the system~\eqref{eq:dELV1} is QMS using
invariants of coalgebraic origin. Our starting point is the following
proposition, which states that we can find two additional invariants through
a direct search:

\begin{proposition}\label{prop:invariantsV1}
    The system~\eqref{eq:dELV1}, generated by the
    Lagrangian~\eqref{eq:lagrangianV1}, also admits the following invariants:
    \begin{subequations}\label{Ts}
        \begin{align}
            J_{1}   &\coloneqq
            \begin{aligned}[t]
	 	  &-\alpha_{+}^{2} \left(\alpha_{+}^{2}+\varkappa -2\right) A_{-} A_{+}+2 \alpha_{+} A_{-} B_{-}-2 \alpha_{+} \left(\alpha_{+}-1\right) \left(\alpha_{+}+1\right) A_{-} B_{+}  
	\\ 
	&+2 \alpha_{+} \left(\alpha_{+}^{2}+\varkappa \right) A_{-} K-2 \alpha_{+} \left(\alpha_{+}-1\right) \left(\alpha_{+}+1\right) A_{+} B_{-}
	+2 A_{+} \alpha_{+} B_{+}
	\\ 
        &+2 \alpha_{+} \left(\alpha_{+}^{2} +\varkappa \right) A_{+} K+B_{-}^{2}
	+\left(-\varkappa  \alpha_{+}^{2} -2 \alpha_{+}^{2}+2\right) B_{+} B_{-}
	+2 \varkappa  B_{-} 
	\\
	&+B_{+}^{2} +2 \varkappa  B_{+} K+\left(\varkappa  \alpha_{+}^{2}+\varkappa^{2}+2 \alpha_{+}^{2}\right) K^{2}+\alpha_{+}^{2} \left(\alpha_{+}^{2}+\varkappa +2\right) K M, 
            \end{aligned}
            \\
             \widehat J_{1} &\coloneqq  
             \begin{aligned}[t]
                \alpha_{+}^{2} A_{+} A_{-}+2 \alpha_{+} A_{-} B_{+}-2 \alpha_{+} A_{-} K+2 \alpha_{+} A_{+}  B_{-}-2\alpha_{+}  A_{+}  K
                \\
                +\left(\varkappa +2\right) B_{-} B_{+} +B_{-} M+B_{+} M
                -\left(\varkappa +2\right) K^{2}
                -\left(\alpha_{+}^{2}+2\right) K M
             \end{aligned}
        \end{align}
    \end{subequations}
\end{proposition}

Then, we prove the following theorem:

\begin{theorem}\label{th:QMS}
    The system~\eqref{eq:dELV1}, generated by the Lagrangian~\eqref{eq:lagrangianV1}, is quasi-maximally superintegrable
    because it admits a set of $2N-2$ functionally independent\ invariants:
    \begin{equation}\label{QMS-set} 
        \mathcal{S}_1 = \left\{ I_1, J_{1},  \widehat  {J}_{1}, C^{[3]}, \dots, C^{[N]}, C_{[3]}, \dots, C_{[N-1]} 
       	 \right\}, 
    \end{equation}  
    of which the first $N$ are in involution.
\end{theorem}

\begin{proof}
    From \Cref{thm:h6}, \Cref{prop:linear}, and \Cref{prop:invariantsV1}, we
    have that the system~\eqref{eq:dELV1} admits the invariants listed in
    equation~\eqref{QMS-set}.
  
    Similarly to the proofs of \Cref{th:SI} and \Cref{th:SI1} it is enough to show
    that:
    \begin{align}
        \rank\left[ \frac{\partial (I_1, J_{1},  \widehat {J}_{1}, C )}{\partial ( \vec q, \vec p)    } \right]= 4.
    \end{align}

    Let us start by proving that:
    \begin{align}
        \rank\left[ \frac{\partial (I_1, J_{1},  \widehat {J}_{1} )}{\partial ( \vec q, \vec p)    } \right]= 3.
    \end{align}
    Observe that the $I_{1}$, $J_{1}$, and $\widehat {J}_{1}$ have coalgebraic
    origin. So, we can consider them as functions of the generators of algebra
    composed with the $N$-dimensional symplectic realization. 
    That is, we can consider the following decomposition of the Jacobian:
    \begin{equation}
        \label{eq:JacDecV1}
        \frac{\partial (I_1, J_{1},  \widehat {J}_{1}  )} {\partial (\vec q, \vec p)} = 
        \frac{\partial (I_1, J_{1},  \widehat {J}_{1}  )} 
        {\partial (A_+, A_-, B_+, B_-, K )} 
        . 
        \frac{\partial(A_+, A_-, B_+, B_-, K )}{\partial (\vec q, \vec p)},
    \end{equation}
    where the second factor is given by \eqref{eq:jacobian-partial-generators-in-realization}. 
    
     The other factor in~\eqref{eq:JacDecV1} has the form 
\begin{equation} \label{Jacobian-abstract} 
	   \frac{\partial (I_1,  J_{1},  \widehat {J}_{1}  )} 
	   {\partial (A_+, A_-, B_+, B_-, K )}  = 
	   \begin{pmatrix}
	   	\alpha_+ & \alpha_+ & 1 & 1 & \varkappa 
	   	\\
	   	\Theta_{-} & \Theta_{+} & \Xi_- & \Xi_+ & \Omega
	   		\\
	   	\Psi_- & \Psi_+  & \Pi_{-} & \Pi_+ & \Upsilon 
	   \end{pmatrix}_{3 \times 5} , 
    \end{equation}
    where 
    \begin{subequations}\label{matrix-coefficients}
    \begin{align}
            \Theta_{\pm} & = 2 \alpha_{+} \left(\alpha_{+}^{2}+\varkappa \right) K-\alpha_{+}^{2} \left(\alpha_{+}^{2}+\varkappa -2\right) A_{\pm}+2 \alpha_{+} B_{\mp}-2 \alpha_{+} \left(\alpha_{+}-1\right) \left(\alpha_{+}+1\right) B_{\pm},
            \\
            \Xi_{\pm} &=  2 \varkappa K  +2 \alpha_{+} A_{\mp} -2 \alpha_{+} \left(\alpha_{+}-1\right) \left(\alpha_{+}+1\right) A_{\pm}+2 B_{\mp}+\left(-\varkappa  \alpha_{+}^{2}-2 \alpha_{+}^{2}+2\right) B_{\pm},
            \\
            \Omega & 
            \begin{aligned}[t]
                &=\alpha_{+}^{2} \left(\alpha_{+}^{2}+\varkappa +2\right) M+\left(2 \varkappa  \alpha_{+}^{2}+2 \varkappa^{2}+4 \alpha_{+}^{2}\right) K+2 \alpha_{+} \left(\alpha_{+}^{2}+\varkappa \right) A_{-}
                \\
                &+2 \alpha_{+} \left(\alpha_{+}^{2}+\varkappa \right) A_{+}+2 \varkappa ( B_{-}+ B_{+}),
            \end{aligned}
    \end{align}
    and 
    \begin{align}
            \Psi_\pm &=  \alpha_{+}^{2} A_{\pm } -2\alpha_{+} K +2 \alpha_{+} B_{\pm}, 
            \\
            \Pi_\pm & = M+2 A_{\pm } \alpha_{+}+\left(\varkappa +2\right) B_{\pm},
            \\ 
            \Upsilon & =  -\left(\alpha_{+}^{2}+2\right) M
            -2\left( \varkappa +2\right) K
            -2 \alpha_+ ( A_+  + A_-
            ). 
    \end{align}
    \end{subequations}
    To prove that the matrix~\eqref{Jacobian-abstract} has a maximal rank $3$, it is sufficient to prove that, for instance, the following determinant does not vanish:
    \begin{equation}
    	\Delta \coloneqq \det \begin{pmatrix}
		\alpha_+ & 1 & \varkappa 
		\\ 
		\Theta_+ & \Xi_+ & \Omega 
		\\ 
		\Psi_+ & \Pi_+ & \Upsilon 
	\end{pmatrix} \ne 0. 
    \end{equation}
    By continuity of the determinant, we can prove this fact for $\varkappa=0$
    without loss of generality. Expanding the calculations, one arrives at 
    \begin{equation}
            \Delta \big|_{\varkappa = 0} =  
             - \alpha_+^2 (2 A_- + \alpha_+ M) 
             \left( \alpha_+^2 M 
             + 2 \alpha_+ (A_+ + A_-) + 2M + 4K
             \right), 
    \end{equation}
    which is not zero for non-vanishing $\alpha_+$. Therefore, 
    \begin{equation}
            \rank \left[ \frac{\partial ( I_1, J_{1},  \widehat {J}_{1}  )} 
               {\partial (A_+, A_-, B_+, B_-, K )}  \right] =3.
    \end{equation}

    Since in general for matrix products it holds
    $\rank (AB)= \min ( \rank(A), \rank(B))$, if $B$ has maximal rank,
    we can conclude that
    \begin{equation}\label{full-rank}
        \rank \left(   \frac{\partial ( I_1, J_{1},  \widehat {J}_{1}  )} {\partial (\vec q, \vec p)} \right)  
        =3.  
    \end{equation}
    Finally, the ``full'' Jacobi matrix $\partial (I_1, J_1, \widehat J_1, C^{[N]} )/ \partial (\vec q, \vec p) $  in the
    symplectic realization~\eqref{eq:h6reaN} can be rewritten as 
\begin{equation}\label{jacobian-realization-zeta}
		\frac{\partial ( I_1, J_{1},  \widehat {J}_{1} , C  )}{\partial ( \vec q, \vec p)    }  
		= \begin{pmatrix}
			 \vartheta_{1,1} & \dots &   \vartheta_{1,2N} 
			 \\ 
			  \vartheta_{2, 1} & \dots &  \vartheta_{2, 2N}
			 \\ 
			  \vartheta_{3,1} & \dots &  \vartheta_{3, 2N} 
			 \\ 
			   \frac{\partial C  }{\partial q_1 }  & \dots  &
  	 		  \frac{\partial  C }{\partial p_N } &
		\end{pmatrix}, 
\end{equation}
where the matrix elements are defined to be  
\begin{equation}
	 \vartheta_{i,j} \coloneqq 
	\left[    \frac{\partial ( I_1,  J_{1},  \widehat {J}_{1} )} 
	   {\partial (A_+, A_-, B_+, B_-, K )} 
	   . 
	   \frac{\partial(A_+, A_-, B_+, B_-, K )}{\partial (\vec q, \vec p)} \right]_{i,j}, 
	   \qquad 
	   i=1,2,3; \quad j =1,\dots, 2N, 
\end{equation}
All the matrix elements $ \vartheta_{i,j} $ are again polynomials with respect to
$\alpha_+$ of some degree by construction. In contrast, the left Casimir (and
hence the components of its gradient $	\nabla_{ ( \vec q, \vec p)} C^{[N]}$ )
do not depend on $\alpha_+$, hence all the elements $    {\partial C^{[N]}
}/{\partial q_i }  $ and $    {\partial C^{[N]} }/{\partial p_j}  $ are
polynomials of degree $0$ in $\alpha_+$. Therefore, the 4th row of
\eqref{jacobian-realization-zeta} is linearly independent of all other rows.
However, the linear independence of rows 1-3 from each other was already
established, since \eqref{full-rank} holds.  Therefore,  \begin{equation} \rank 	 
\left( \frac{\partial ( I_1, J_{1},  \widehat {J}_{1}, C^{[N]} )}{\partial ( \vec q, \vec p)    }  \right) =4,
\end{equation}
which ends the proof.
\end{proof}

\begin{remark}
    We remark that under the scaling \eqref{eq:coeff_scalingV1} the invariants
    $J_{1}$ and $\widehat{J}_{1}$ have the following continuum limits:
    \begin{subequations}\label{eq:ioms-cont-V1}
        \begin{align}
            J_{1} &=\vl^{4} +  \vl^2  \mathcal{H}_{1} + \mathcal{O}(h^{3}), 
            \\ 
            \widehat J_1  & =
            \vl^{4} + h^{2} \vl^2 \left[ 2 \mathcal{H}_{1}
            - \frac{\omega^2 \vl^2 }{4}
            \right]
            + \mathcal{O}(h^{3}). 
        \end{align}
    \end{subequations}
    That is, in the continuum limit, they collapse to the Hamiltonian of the
    system~\eqref{eq:cont-Hamiltonian-V1}. This is not surprising, as it
    happens in many cases, see for instance the examples in~\cite[\S
    4]{Gubbiotti_lagr_add}. In general, this does not give any precise insight
    on the integrability properties of the continuum limits, which must be
    discussed separately.
    \label{rmk:contlimI2J2}
\end{remark}

\printbibliography

\end{document}